\documentclass[twocolumn,10pt]{autart}  
\usepackage{graphicx}         
\usepackage{amsmath}
\usepackage{enumerate}
\usepackage{graphicx}       
\usepackage{tikz}
\usepackage{apacite}
\usepackage{natbib}
\usepackage{amsfonts}

\usepackage{amssymb}
\usepackage{subfigure}
\usepackage{hyperref}
\hypersetup{
 colorlinks = true,
 citecolor = cyan,
}
\usepackage{wrapfig}
\newenvironment{proof}{{\indent \indent \bf Proof.}}{\hfill $\blacksquare$\par}

\newtheorem{theorem}{Theorem}
\newtheorem{lemma}{Lemma}
\newtheorem{remark}{Remark}
\newtheorem{proposition}{Proposition}

\newtheorem{corollary}{Corollary}
\newtheorem{example}{Example}

\begin{document}

\begin{frontmatter}

\title{Controllability of networked multiagent systems based on linearized Turing's model\thanksref{footnoteinfo}} % Title, preferably not more 

\thanks[footnoteinfo]{This work was supported by Natural Science Foundation of China under Grant T2293772, National Key R\&D Program of China under Grant 2018YFA0703800, Shanghai Municipal Science and Technology Major Project (No.2021SHZDZX0103), the Strategic Priority Research Program of Chinese Academy of Sciences under Grant No. XDA27000000, and National Science Foundation of Shandong Province (ZR2020ZD26). This paper was not presented at any IFAC meeting. Corresponding author Zhixin.~Liu.}

\author[amss]{Tianhao Li}\ead{litianhao@amss.ac.cn},    % Add the 
\author[amss]{Ruichang Zhang}\ead{zrc20@amss.ac.cn}, 
\author[amss]{Zhixin Liu}\ead{lzx@amss.ac.cn},    % e-mail address 
\author[fudan]{Zhuo Zou}\ead{zhuo@fudan.edu.cn},
\author[kth]{Xiaoming Hu}\ead{hu@kth.se}  % (ead) as shown

\address[amss]{Key Laboratory of Systems and Control, Academy of Mathematics and Systems Science,
        Chinese Academy of Sciences, and School of Mathematical Sciences, University of Chinese Academy of Sciences, Beijing 100190, P.~R.~China}  % Please supply    
\address[fudan]{Fudan University, Shanghai 200433, P.~R.~China}
\address[kth]{Optimization and Systems Theory, KTH Royal Institute of Technology, Stockholm 10044, Sweden}             % full addresses
            % here.

                             % keyword list or with the 
                                          % help of the Automatica 
                                          % keyword wizard

\begin{abstract}                          % Abstract of not more than 200 words.
Turing's model has been widely used to explain how simple, uniform structures can give rise to complex, patterned structures during the development of organisms. However, it is very hard to establish rigorous theoretical results for the dynamic evolution behavior of Turing's model since it is described by nonlinear partial differential equations. We focus on controllability of Turing's model by linearization and spatial discretization. This linearized model is a networked system whose agents are second order linear systems and these agents interact with each other by Laplacian dynamics on a graph. A control signal can be added to agents of choice. Under mild conditions on the parameters of the linearized Turing's model, we prove the equivalence between controllability of the linearized Turing's model and controllability of a Laplace dynamic system with agents of first order dynamics. When the graph is a grid graph or a cylinder grid graph, we then give precisely the minimal number of control nodes and a corresponding control node set such that the Laplace dynamic systems on these graphs with agents of first order dynamics are controllable.
\end{abstract}

\end{frontmatter}

\section{Introduction}

Alan Turing proposed a model in his seminal paper \cite{Turing1952} to explain how the formation of various inhomogeneous creature patterns, such as pigmentation in animals, branching in trees and skeletal structures, arises from homogeneous embryo \cite{Maini2012,Gierer1972}. Experimental research indicates that Turing's model plays an important role in biological pattern formation, such as the formation of murine hair follicle spacing \cite{Sick2006} and digit patterning \cite{Sheth2012,Raspopovic2014}. It is a nonlinear partial differential equation with a diffusion term and a nonlinear reaction term, which has given rise to a lot of inhomogeneous patterns in simulations \cite{Maini2012}. It is thus interesting to understand the self-organizing mechanism for creating complex patterns behind Turing's model. However, this self-organizing mechanism is so complicated that it would take a long time effort to fully understand it. Thus, in this work we study as a first step whether it is possible to create patterns by using external interventions instead. This idea coincides with the controllability problem in control theory when the external intervention is viewed as control. It is hard to study controllability of  the original Turing's model due to its nonlinearity and dimensionality. Therefore, in this paper, we study controllability of a networked linear system derived from Turing's model by linearization and spatial discretization. 

Recently, some researchers focus on the controllability of multiagent systems whose agents have high order dynamics \cite{Hao2019b,Zhang2017,Trumpf2019,Wang2016}. In \cite{Wang2016}, a necessary and sufficient condition for a general networked multi-input/multi-output (MIMO) system to be controllable is established in terms of two algebraic matrix equations. In \cite{Hao2018} and \cite{Hao2019b}, based on the eigenvalues and eigenvectors of diagonalizable topology matrices, necessary and sufficient conditions for a networked MIMO system to be controllable are given. Despite the fact that necessary and sufficient conditions for controllability of very general multiagent systems are given in the literature, verification of these conditions remains difficult.

Most research on controllability of multiagent systems of first order dynamics focuses on the Laplace dynamic systems, namely control systems whose system matrix is the Laplacian matrix of a graph. The network controllability problem of a Laplace dynamic system was introduced in \cite{Tanner2004}, where the eigenstructure of system matrix is used to characterize network controllability. Some progresses have been made for the network controllability in a general graph framework \cite{Sun2017,Aguilar2015,Yazıcıoğlu2016,Commault2013,Ji2017,Godsil2012,Rahmani2009}. For example, Godsil in \cite{Godsil2012} presented necessary and sufficient conditions for network controllability of unweighted adjacency dynamic system with a single node by concepts in graph theory. But, this work did not consider the case of multiple control nodes. In \cite{Rahmani2009}, a necessary condition for networked systems with single control node to be controllable is given by using graph symmetry, and then this result is extended to the multiple control node case by introducing the approach of network equitable partitions. However, sufficient conditions to  guarantee that a networked system is controllable in this work are not given. Research on the controllablity of networked system in general graph framework indicates that, without any given structure on the graph, it is hard to give a practically useful criterion for network controllability. Thus, some researchers are interested in controllability of networked systems on graphs of special structures, such as composite network constructed by Cartesian Products \cite{She2021,Chapman2014,Hao2019b}, Kronecker product \cite{Hao2019a}, join and union operations \cite{Mousavi2021}. In \cite{Chapman2014,Hao2019a}, the authors give some necessary and sufficient conditions for network controllability with multiple control nodes for Cartesian product network and Kronecker product network. In \cite{Mousavi2021}, Mousavi, Haeri and Mesbahi present a necessary and sufficient condition for the controllability of cographs, and provide an efficient method for selecting a control node set of minimal number for cographs to be controllable. These results are useful and practical for controllability of cographs. Nevertheless, a lot of graphs in reality are not cographs, such as almost all path graphs, cycle graphs and grid graphs. 

There is also research on controllability of more concrete networks, such as path graph \cite{She2020,Parlangeli2012}, cycle \cite{She2020,Parlangeli2012}, multichain \cite{Cao2013}, tree \cite{Ji2012}, grid graph \cite{Notarstefano2013} and circulant graph \cite{Nabi-Abdolyousefi2013}. In \cite{Parlangeli2012}, a necessary and sufficient condition based on simple relations from number theory for Laplace dynamic system on path graph to be controllable is given, by which network controllability of the Laplace dynamic system on the path graph is completely solved. Similar condition is also given for cycle graph in \cite{Parlangeli2012}. In \cite{Notarstefano2013}, Notarstefano and Parlangeli give a necessary and sufficient condition for controllability of the Laplace dynamic system on simple grid graphs, namely grid graphs whose Laplacian matrices have only simple eigenvalues. Then, the eigenstructure of Laplacian matrices of general grid graphs is analyzed. Yet, they did not give efficient method for the case when the grid graph is not simple. 

There is also research that does not aim at determining whether a networked system is controllable for given control node set, rather at finding a control node set with the minimal number of control nodes which assure controllability of the networked system \cite{Yuan2013,Olshevsky2014,Mousavi2021,Nabi-Abdolyousefi2013}. This problem is called ``Minimal controllability problem" and is proved to be NP-hard in \cite{Olshevsky2014}. In \cite{Yuan2013}, it is proposed that the minimal number of control nodes to ensure controllability of a networked system is equal to the maximum geometric multiplicity of Laplacian matrix for general network topology and link weights, and a method for selecting minimal set of control node to achieve full controllability of networks is obtained, provided some limitation on the network topology or link weights is given. In \cite{Nabi-Abdolyousefi2013}, the minimal number of control nodes for circulant networks is proved to be the maximum multiplicity of the Laplacian matrix. To the best of our knowledge, none of the results on ``Minimal controllability problem" in previous research works when the graph is a grid graph or cylinder grid graph.

In this paper, by using analytical methods such as PBH Criterion, Kalman Criterion and the solutions of trigonometric diophantine equations \cite{Conway1976,Wlodarski1969}, we study controllability of a linearized Turing's model, and give the minimal number of control nodes and a corresponding control node set for the linearized Turing's model to be controllable when the graph is chosen as a grid graph or cylinder grid graph.

The contribution of this paper is as follows. For a linearized Turing's model, which is a networked system with agents of second order dynamics, under mild conditions, we prove that its controllability is equivalent to network controllability of a Laplace dynamic system with agents of first order dynamics.  As a consequence, the controllability problem of the linearized Turing's model is significantly simplified. Taking advantage of the special structure of the linearized Turing's model, our conditions are concise and easy to verify compared to the previous work on controllability of general networked MIMO system \cite{Wang2016,Hao2018}; \cite{Hao2019b}. Furthermore, for Laplace dynamic systems on a grid graph, by carefully analyzing the eigenvalues and eigenvectors of Laplacian matrices of grid graphs, and by using the method of trigonometric diophantine equations, we give the precisely minimal number of control nodes and a minimal control nodes set so that the system is controllable. This number only depends on the row number and column number of the grid graph and can be easily computed. In fact, the greatest common divisor of the row number and column number plays an important role in this minimal number. Similar results are given when the base graph is a cylinder grid graph. These results are more general and complete compared to the previous work on controllability of grid graphs \cite{Notarstefano2013} as they did not give any efficient method for the case when the grid graph has multiple eigenvalues. Although we consider only the case when the system matrices of the control systems are the Laplacian matrices of two types of graphs mentioned above, the methods proposed in this paper can be easily applied to obtain similar results when the system matrices of the control systems are adjacency matrices of grid graphs and cylinder grid graphs.

The remainder of this paper is organized as follows. In Section \ref{section2}, we first introduce relevant background material pertaining to graphs, Cartesian products, and Kronecker products, and then introduce the linearized Turing's model and its controllability problem. In Section \ref{section3}, we analyze controllability of the linearized Turing's model. In Section \ref{section4}, controllability of the Laplace dynamic systems on grid graphs and cylinder grid graphs is analyzed. 

\section{Preliminaries and problem formulation}\label{section2}

\subsection{Preliminaries}\label{subsection2.1}

In this paper, $\mathbf{R}$, $\mathbf{C}$ and $\mathbf{N}$ denote the set of real, complex and natural numbers, respectively. $I_{n}=[e^{(n)}_{1},\cdots, e^{(n)}_{n}]$ denotes the $n\times n$ unit matrix where  $e^{(n)}_{j}$ denotes the $j$th column of $I_{n}$. For two matrices $P\in \mathbf{R}^{m\times m}$ and $Q\in \mathbf{R}^{n\times n}$, $P\oplus Q = P\otimes I_{n}+I_{m}\otimes Q$ is called the Kronecker sum of $P$ and $Q$, where $\otimes$ denotes the Kronecker product. If $P$ and $Q$ have (either left or right) eigenvalue-eigenvector pairs $(\lambda_{i},v_{i})$ for $1\leq i \leq p_{1} (p_{1}\leq m)$ and $(\mu_{j},w_{j})$ for $1\leq j\leq p_{2} (p_{2}\leq n)$, respectively, then $P\oplus Q$ has eigenvalue-eigenvector pairs $(\lambda_{i}+\mu_{j},v_{i}\otimes w_{j})$ for $1\leq i \leq p_{1}$ and $1\leq j\leq p_{2}$.

A graph is denoted as $\mathcal{G}=(\mathcal{V},\mathcal{E})$, where $\mathcal{V}=\{\zeta_{i}|1\leq i\leq N\}$ represents the vertex set, $\mathcal{E}\subset \mathcal{V}\times \mathcal{V}$ is the edge set, and $\times$ denotes the Cartesian product of two sets. All the graphs mentioned in this paper are undirected. The Laplacian matrix of a graph is denoted by $L$.

The Cartesian product of graphs can help us construct properties of complex graphs from those of simple graphs. For two graphs $\mathcal{G}_{1}=(\mathcal{V}_{1},\mathcal{E}_{1})$ and $\mathcal{G}_{2}=(\mathcal{V}_{2},\mathcal{E}_{2})$, $\mathcal{G}_{1}\Box \mathcal{G}_{2}=(\mathcal{V},\mathcal{E})$ is called their Cartesian product, where $\mathcal{V}=\mathcal{V}_{1}\times \mathcal{V}_{2}$ is the Cartesian product of sets $\mathcal{V}_{1}$ and $\mathcal{V}_{2}$, and $\mathcal{E}=\{((\zeta_{i},\eta_{i}),(\zeta_{j},\eta_{j}))\in (\mathcal{V}_{1}\times \mathcal{V}_{2})\times(\mathcal{V}_{1}\times \mathcal{V}_{2})|\zeta_{i}=\zeta_{j}, (\eta_{i},\eta_{j})\in \mathcal{E}_{2};\ or \ (\zeta_{i},\zeta_{j})\in \mathcal{E}_{1}, \eta_{i}=\eta_{j}\}$. The Laplacian matrix of $\mathcal{G}_{1}\Box \mathcal{G}_{2}$ are $L_{1}\oplus L_{2}$, where $L_{k}(k=1,2)$ are Laplacian matrix of $\mathcal{G}_{k}$.

%  and $A(m,n)$, then we have $L(m,n)=L_{\mathcal{P}}(m)\oplus L_{\mathcal{P}}(n)$ and $A(m,n)=A_{\mathcal{P}}(m)\oplus A_{\mathcal{P}}(n)$, where $L_{\mathcal{P}}(m)$ and $A_{\mathcal{P}}(m)$ are the Laplacian matrix and adjacency matrix of $\mathcal{P}_{m}$.

In the following, we introduce several commonly used graphs. For a graph $\mathcal{G}=(\mathcal{V},\mathcal{E})$, if $\mathcal{V}=\{\zeta_{i}|1\leq i\leq n\}$ and $\mathcal{E}=\{(\zeta_{i},\zeta_{i+1})|1\leq i\leq n-1\}$, then it is called a path graph and denoted as $\mathcal{P}_{n}$. The Cartesian product of two path graphs $\mathcal{P}_{m}$ and $\mathcal{P}_{n}$ is called a grid graph (see Fig.~\ref{fig1_1}), and denoted by $\mathcal{L}_{m,n}=\mathcal{P}_{m}\Box \mathcal{P}_{n}$. A graph $\mathcal{G}=(\mathcal{V},\mathcal{E})$ is called a cycle graph, if $\mathcal{V}=\{\zeta_{i}|1\leq i\leq n\}$ and $\mathcal{E}=\{(\zeta_{i},\zeta_{i+1})|1\leq i\leq n-1\}\cup \{(\zeta_{n},\zeta_{1})\}$. Denote a cycle graph with $m$ nodes as $\mathcal{D}_{m}$. The Cartesian product of a cycle graph $\mathcal{D}_{m}$ and a path graph $\mathcal{P}_{n} $ is called a cylinder grid graph (see Fig.~\ref{fig1_2}), and denoted by $\mathcal{C}_{m,n}=\mathcal{D}_{m} \Box \mathcal{P}_{n}$.

\begin{figure}
\begin{center}
% \subfigtopskip=2pt
% \subfigbottomskip=2pt
% \subfigcapskip=-5pt
\subfigure[A $9\times 9$ grid graph]{
    \label{fig1_1}
    \includegraphics[width=0.45\linewidth]{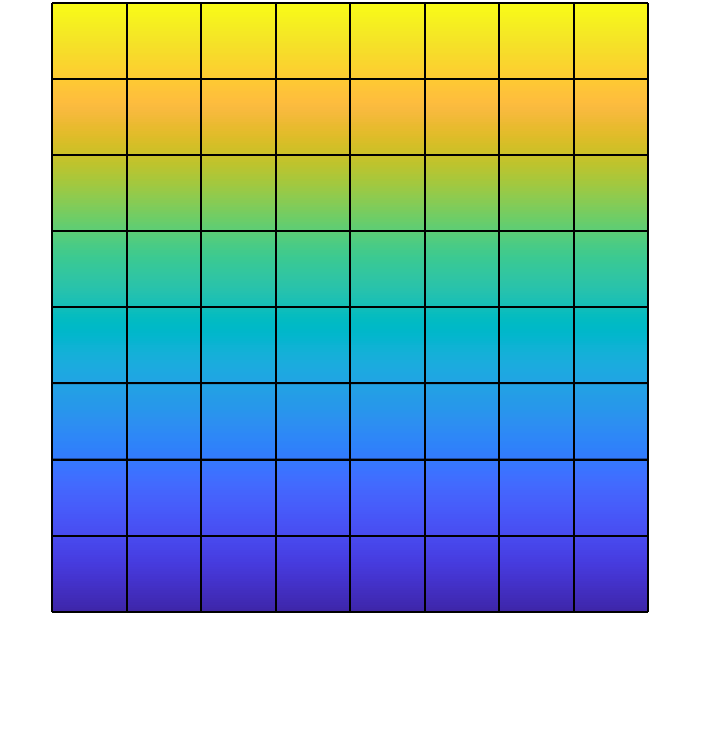}}
% \quad
\subfigure[An $8\times 9$ cylinder grid graph]{
    \label{fig1_2}
    \includegraphics[width=0.45\linewidth]{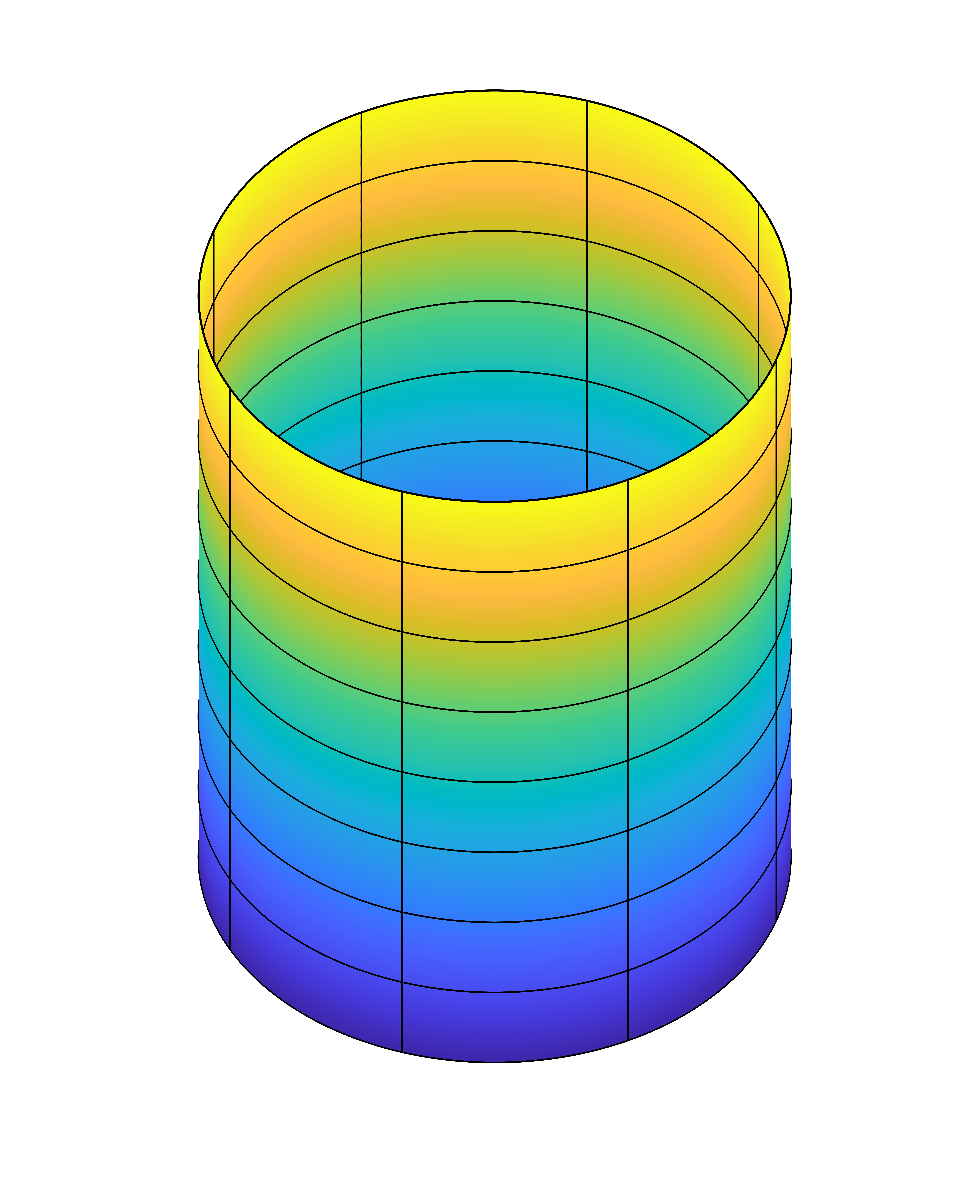}}
\caption{The figures of a grid graph and a cylinder grid graph.}
\label{fig1}
% \includegraphics[width=7.0cm]{grid2.eps} 
%   \caption{The figure of a $9\times 9$ grid graph}
%   \label{fig1}                            % Size the figures 
\end{center}                                 % accordingly.
\end{figure}

% \begin{figure}
% \begin{center}
% \includegraphics[height=8.4cm]{cylinder3.eps} 
%   \caption{The figure of a $8\times 9$ cylinder grid graph}
%   \label{fig2}                            % Size the figures 
% \end{center}                                 % accordingly.
% \end{figure}

% \begin{figure}
% \begin{center}
% \includegraphics[width=8.4cm]{torus3_2.eps} 
%   \caption{The figure of a $8\times 8$ torus grid graph}
%   \label{fig3}                            % Size the figures 
% \end{center}                                 % accordingly.
% \end{figure}

\subsection{Problem formulation}
The Turing's model is proposed to explain the process of biological pattern formation(cf., \cite{Turing1952,Maini2012}). The dynamics of the system is described by
\begin{equation}
  \nonumber
  \left\{
    \begin{array}{c}
    \frac{\partial {z}_{1}(t,x,y)}{\partial t}=\nu_{1} \Delta z_{1}+f(z_{1},z_{2}),\\
    \frac{\partial {z}_{2}(t,x,y)}{\partial t}=\nu_{2} \Delta z_{2}+g(z_{1},z_{2}),
    \end{array}
    \right.
\end{equation}
where $t\in [0,+\infty)$, $(x,y)\in \mathbf{D}$ with $\mathbf{D}$ being a bounded region contained in $\mathbf{R}^{2}$, $z_{i}\in \mathbf{R}$,  $f,g: \mathbf{R}^{2}\rightarrow \mathbf{R}$ are two differential functions, $\Delta = \frac{\partial ^2}{\partial x^2} + \frac{\partial ^2}{\partial y^2}$ is the Laplace operator in $\mathbf{R}^{2}$, and $\nu_{1},\nu_{2}$ are two positive constants. Appropriate boundary and initial conditions are applied to close the system. In the Turing's model, ${z}_{i}(t,x,y)$ represents concentrations of chemical material $\mathbf{M}_{i}$ at time $t$ and location $(x,y)$. $\nu_{1} \Delta z_{1}$ and $\nu_{2} \Delta z_{2}$ represent  diffusion rates of chemical materials $\mathbf{M}_{1}$ and $\mathbf{M}_{2}$, respectively, and  $\nu_{1}$ and $\nu_{2}$ are diffusion coefficients. $f(z_{1},z_{2})$ and $g(z_{1},z_{2})$ are the production rates or consumption rates of $\mathbf{M}_{1}$ and $\mathbf{M}_{2}$ by chemical reaction.

% We simplify the model by linearization and spatial discretization by replacing function $f(z_{1},z_{2})$ and $g(z_{1},z_{2})$ by their linear part $az_{1}+bz_{2}$ and $cz_{1}+dz_{2}$. We obtain a linear PDE. Then we further simplify the model by spatial discretization. 
This model is described by nonlinear partial differential equations, whose dynamics are hard to analyze. To simplify this model, we first discretize the region $\mathbf{D}$ into a graph of $N$ vertexes. Then, by linearization and spatial discretization (sampling on vertexes of the graph in $\mathbf{D}$), we get the following linear second-order multi-agent system,
\begin{equation}
  \nonumber
  \left\{
    \begin{array}{c}
    \dot{z}_{1}(t)=-\nu_{1} Lz_{1}+az_{1}+bz_{2},\\
    \dot{z}_{2}(t)=-\nu_{2} Lz_{2}+cz_{1}+dz_{2},
    \end{array}
    \right.
\end{equation}
where $L$ is the Laplacian matrix of the graph, $z_{i}(t)=[z_{i}^{(1)},\cdots,z_{i}^{(N)}]\in \mathbf{R}^{N}$ denotes the concentration of $M_{i}$ at sampling points for $i=1,2$. Each vertex $j\ (1\leq j\leq N)$, also called agent,  has two states $z^{(j)}_{1}$ and $z^{(j)}_{2}$. Few patterns will emerge if the above linear system evolves in a self-organization manner. We are interested in whether it is possible to generate desired patterns by external intervention. Thus, we choose $\kappa+\tau$ states $z^{(i_{1})}_{1},\cdots,z^{(i_{\kappa})}_{1},z^{(j_{1})}_{2},\cdots,z^{(j_{\tau})}_{2}$ to add control. By this way, we can obtain the following controlled linearized Turing's model,  
% In this paper, we study the controllability of the following networked system,
\begin{equation}
  \label{eq1}
  \left\{
    \begin{array}{c}
    \dot{z}_{1}(t)=-\nu_{1} Lz_{1}+az_{1}+bz_{2}+Bu_{1},\\
    \dot{z}_{2}(t)=-\nu_{2} Lz_{2}+cz_{1}+dz_{2}+Cu_{2},
    \end{array}
    \right.
\end{equation}
where $L$ is the Laplacian matrix of a graph, $\nu_{1},\nu_{2}>0$, $a,b,c,d\in \mathbf{R}$, $z_{i}(t)\in \mathbf{R}^{N}$ for $i=1,2$ and $t\in [0,+\infty)$, $B=[e^{(N)}_{i_{1}},\cdots,e^{(N)}_{i_{\kappa}}]$ and $C=[e^{(N)}_{j_{1}},\cdots,e^{(N)}_{j_{\tau}}]$ ( $1\leq i_{k},j_{l}\leq N$ for $1\leq k\leq \kappa$ and $1\leq l\leq \tau$) are control matrices, $e^{(N)}_{j}$ denotes the $j$th column of $N\times N$ unit matrix, $u_{1}\in \mathbf{R}^{\kappa}$ and $u_{2}\in \mathbf{R}^{\tau}$ are the control inputs. 

% \begin{remark}
%   The networked system (\ref{eq1}) comes from the Turing's model, which is one of the best-konwn theoretical model to explain biological pattern formation, such as pigmentation in animals, branching in trees and skeletal structures\cite{Turing1952,Maini2012,Gierer1972}. Turing's model is the following system,
% \begin{equation}
%   \nonumber
%   \left\{
%     \begin{array}{c}
%     \frac{\partial {z}_{1}(t,x,y)}{\partial t}=\nu_{1} \Delta z_{1}+f(z_{1},z_{2}),\\
%     \frac{\partial {z}_{2}(t,x,y)}{\partial t}=\nu_{2} \Delta z_{2}+g(z_{1},z_{2}),
%     \end{array}
%     \right.
% \end{equation}
% where $z_{i}:[0,+\infty)\times \mathbf{R}^{2}\rightarrow \mathbf{R}$,  $f,g: \mathbf{R}^{2}\rightarrow \mathbf{R}$ are differential functions, $\Delta$ is the Laplace operator in $\mathbf{R}^{2}$ and $\nu_{1},\nu_{2} \in (0,+\infty)$. We first simplify the Turing's model by linearization and spatial discretization, and then we choose $r+t$ state $z^{(i_{1})}_{1},\cdots,z^{(i_{r})}_{1},z^{(j_{1})}_{2},\cdots,z^{(j_{t})}_{2}$ to add control nodes. Finally, we obtain the networked system (\ref{eq1}).
% \end{remark}
 If system (\ref{eq1}) is controllable, then the desired behaviors can be achieved by designing control inputs.  In this paper, we will investigate  how to choose control matrices $B$ and $C$ to achieve controllability of the system. 

In the controlled linearized Turing's model (\ref{eq1}), by choosing some states instead of vertexes to add control,  we can reduce the amount of controlled states to achieve controllability of the system. We illustrate this point by an example. 
\begin{example}
Let the parameters in system (\ref{eq1}) be taken as $\nu_{1} = \nu_{2} =1$, $a=1$, $b=2$, $c = 3$, $d=4$. The number of vertexes is  $N=3$, and the graph is taken as a path graph with the corresponding Laplacian matrix   
\begin{equation}
    \nonumber
    L = \begin{bmatrix}
    1 & -1 & 0\\
    -1 & 2 & -1\\
    0 & -1 & 1
    \end{bmatrix}.
\end{equation}
If we choose vertexes to add control, then we can achieve controllability of the system by  controlling the first vertex. For this case, the control matrices are $B=C=e^{(3)}_{1}$, and two states $z^{(1)}_{1}$ and $z^{(1)}_{2}$ are controlled. On the other hand, if we choose states to add control, then we only need to choose the first state of the first vertex, i.e. $z^{(1)}_{1}$, to achieve controllability of the system. The control matrices are $B=e^{(3)}_{1}$ and $C=0$.
\end{example}

\section{Controllability of the linearized Turing's model}\label{section3}

Denote
\begin{equation}\nonumber
    A = \begin{bmatrix} -\nu_{1} L+aI & bI \\ cI & -\nu_{2} L+dI \end{bmatrix},
\end{equation}
\begin{equation}\nonumber
    \tilde{B} = \begin{bmatrix} B & 0_{N\times \tau} \\ 0_{N\times \kappa} & C \end{bmatrix}.
\end{equation}
Then (\ref{eq1}) can be written as
\begin{equation}\nonumber
    \dot{z}=Az+\tilde{B}u,
\end{equation}
where $z = [z_{1}^{T},z_{2}^{T}]^{T}$ and $u=[u_{1}^{T},u_{2}^{T}]^{T}$.

% To analyze controllability of system (\ref{eq1}), a well-known conclusion, Popov-Belevitch-Hautus criterion is needed.

% \begin{lemma}[PBH Criterion]
%     The following statements are equivalent:
%     \begin{itemize}
%         \item[(1)]  $(L,B)$ is uncontrollable; 
%         \item[(2)]  There exists $v^{T}$, a left eigenvector of $L$, such that $v^{T}B=0$;
%         \item[(3)]  There exists $s$, an eigenvalue of $L$, such that $rank(\begin{bmatrix} sI-L & B \end{bmatrix})<N$.
%     \end{itemize}
% \end{lemma}

% \begin{corollary}
%     If $L$ has an eigenvalue of $k$ multiplicities, $rank(B)<k$, then $(L,B)$ is uncontrollable.
% \end{corollary}

We first give a necessary condition for the controllability of $(A,\tilde{B})$.
\begin{lemma}\label{lemma1}
If $(A,\tilde{B})$ is controllable, then $(L, \begin{bmatrix} B & C \end{bmatrix})$ is controllable.
\end{lemma}
\begin{proof}
If $(L, \begin{bmatrix} B & C \end{bmatrix})$ is uncontrollable, then by PBH criterion (\cite{Bacciotti2019}), there exist $\lambda \in \mathbf{C}$ and a nonzero vector $v \in \mathbf{C}^{N}$ such that $v^{T}L=\lambda v^{T}$, and $v^{T}B=v^{T}C=0$. By calculation, we can derive that  $[v^{T},0_{1\times N}]A^{k}\tilde{B}=0$ for $k=0,1,...,N-1$, which indicates that $(A,\tilde{B})$ is uncontrollable. By reduction to absurdity, the lemma is proved.
\end{proof}

In the following, we consider the sufficient conditions for the controllability of $(A,\tilde{B})$. Since both the analysis and results for the  controllability of $(A,\tilde{B})$ are different for the cases of $C=B$ and $C\neq B$, we will study these two cases separately. We first  provide a necessary and sufficient condition for the controllability of $(A,\tilde{B})$ when $C=B$.

\begin{proposition}\label{proposition1}
If $\nu_{1} \neq 0, \nu_{2} \neq 0$ and $C=B$, then $(A, \tilde{B})$ is controllable if and only if $(L,B)$ is controllable.
\end{proposition}
\begin{proof}
By Lemma \ref{lemma1}, we have the necessity. Now, we consider the sufficiency. Denote
\begin{equation}\nonumber
\begin{aligned}
    & B_{1}=\begin{bmatrix} B \\ 0_{N\times \kappa} \end{bmatrix}, B_{2}=\begin{bmatrix} 0_{N\times \kappa} \\ B \end{bmatrix},\\
    & F = \begin{bmatrix} aI & bI \\ cI & dI \end{bmatrix}, \bar{L}=\begin{bmatrix} L & 0_{N\times \kappa} \\ 0_{N\times \kappa} & L  \end{bmatrix}.
\end{aligned}
\end{equation}
We have for $k\in \mathbf{N}$,
\begin{align}
    A \bar{L}^{k}B_{1}=&-\nu_{1} \bar{L}^{k+1}B_{1}+F\bar{L}^{k}B_{1} \nonumber \\
    &=-\nu_{1} \bar{L}^{k+1}B_{1}+a\bar{L}^{k}B_{1}+c\bar{L}^{k}B_{2},\label{eq2}
\end{align}
and
\begin{align}
    A \bar{L}^{k}B_{2}&=-\nu_{2} \bar{L}^{k+1}B_{2}+F\bar{L}^{k}B_{2} \nonumber\\
    &=-\nu_{2} \bar{L}^{k+1}B_{2}+b\bar{L}^{k}B_{1}+d\bar{L}^{k}B_{2}.\label{eq3}
\end{align}

We claim that for $k\in \mathbf{Z}^{+}$,
\begin{equation}
\label{eq4}
    A^{k}B_{1} = (-\nu_{1})^{k}\bar{L}^{k}B_{1}+ \sum_{i=1}^{k}(p_{ik}\bar{L}^{i-1}B_{1}+q_{ik}\bar{L}^{i-1}B_{2}),
\end{equation}
\begin{equation}
\label{eq5}
    A^{k}B_{2} = (-\nu_{2})^{k}\bar{L}^{k}B_{2}+ \sum_{i=1}^{k}(r_{ik}\bar{L}^{i-1}B_{1}+t_{ik}\bar{L}^{i-1}B_{2}),
\end{equation}
where $p_{ik}$, $q_{ik}$, $r_{ik}$ and $t_{ik}$ are some proper real numbers. We prove equations (\ref{eq4}) and (\ref{eq5}) by induction. It is clear that the equations hold for $k=1$. We assume that equations (\ref{eq4}) and (\ref{eq5}) hold for $k \leq m$. Then for $k=m+1$, by (\ref{eq2}) and (\ref{eq3}), we have
\begin{equation}
\begin{aligned}\nonumber
    &A^{m+1}B_{1} = A(A^{m}B_{1})\\
    &=(-\nu_{1})^{m}A\bar{L}^{m}B_{1}+\sum_{i=1}^{m}(p_{im}A\bar{L}^{i-1}B_{1}+q_{im}A\bar{L}^{i-1}B_{2})\\
    &= (-\nu_{1})^{m+1}\bar{L}^{m+1}B_{1}\\
    &+\sum_{i=1}^{m+1}(p_{i,m+1}\bar{L}^{i-1}B_{1}+q_{i,m+1}\bar{L}^{i-1}B_{2}),
\end{aligned}
\end{equation}
\begin{equation}
\begin{aligned}\nonumber
    &A^{m+1}B_{2} = A(A^{m}B_{2})\\
    &=(-\nu_{2})^{m}A\bar{L}^{m}B_{2}+\sum_{i=1}^{m}(r_{im}A\bar{L}^{i-1}B_{1}+t_{im}A\bar{L}^{i-1}B_{2})\\
    &= (-\nu_{2})^{m+1}\bar{L}^{m+1}B_{2}\\
    &+\sum_{i=1}^{m+1}(r_{i,m+1}\bar{L}^{i-1}B_{1}+t_{i,m+1}\bar{L}^{i-1}B_{2}),
\end{aligned}
\end{equation}
where 
\begin{equation}\nonumber
\begin{aligned}
    &p_{m+1,m+1}=(-\nu_{1})^{m}a-\nu_{1} p_{m,m},\\
    &p_{i,m+1} = ap_{i,m}+b q_{i,m}-\nu_{1} p_{i-1,m},\\
    &p_{1,m+1} = ap_{1,m}+b q_{1,m},\\
    &q_{m+1,m+1}=(-\nu_{1})^{m}c-\nu_{2} q_{m,m},\\
    &q_{i,m+1} = cp_{i,m}+d q_{i,m}-\nu_{2} q_{i-1,m},\\
    &q_{1,m+1} = cp_{1,m}+d q_{1,m},\\
    &r_{m+1,m+1}=(-\nu_{2})^{m}b-\nu_{1} r_{m,m},\\
    &r_{i,m+1} = ar_{i,m}+b t_{i,m}-\nu_{1} r_{i-1,m},\\
    &r_{1,m+1} = ar_{1,m}+b t_{1,m},\\
    &t_{m+1,m+1}=(-\nu_{2})^{m}d-\nu_{2} t_{m,m},\\
    &t_{i,m+1} = cr_{i,m}+d t_{i,m}-\nu_{2} t_{i-1,m},\\
    &t_{1,m+1} = cr_{1,m}+d t_{1,m},   
\end{aligned}
\end{equation}
for $2\leq i\leq m$. Thus, equations (\ref{eq4}) and (\ref{eq5}) hold.
By  (\ref{eq4}) and (\ref{eq5}), we have
\begin{equation}
\begin{aligned}\nonumber
    &[B_{1},B_{2},AB_{1},AB_{2},\cdots,A^{2N-1}B_{1},A^{2N-1}B_{2}]\\
    &=[B_{1},B_{2},\bar{L}B_{1},\bar{L}B_{2},\cdots,\bar{L}^{2N-1}B_{1},\bar{L}^{2N-1}B_{2}]Q,
\end{aligned}
\end{equation}
where $Q$ is a $4N\kappa \times 4N\kappa$ invertible matrix.
Thus we have
\begin{equation}
\begin{aligned}\nonumber
    &rank([B_{1},B_{2},AB_{1},AB_{2},\cdots,A^{2N-1}B_{1},A^{2N-1}B_{2}])\\
    &=rank[B_{1},B_{2},\bar{L}B_{1},\bar{L}B_{2},\cdots,\bar{L}^{2N-1}B_{1},\bar{L}^{2N-1}B_{2}]\\
    &=2\times rank([B,LB,\cdots,L^{N-1}B])\\
    &=2N,
\end{aligned}
\end{equation}
which indicates that $(A,\tilde{B})$ is controllable.
\end{proof}

Next, we consider the case of $C\neq B$. Since the equations (\ref{eq2}) and (\ref{eq3}) do not hold when $C\neq B$, the method  used in Proposition \ref{proposition1} is not applicable for such a case.  We focus on analyzing the relationship of eigenvalues and eigenvectors  between  matrices $A$ and $L$. According to the definition of $A$, if $L$ has an eigenvalue $\lambda$ with left eigenvector $v^{T}$, then $A$ has two eigenvalues $s_i(\lambda)(i=1,2)$ and two left eigenvectors $[v^{T},k_i(\lambda)v^{T}] (i=1,2)$, where $s_i(\lambda)(i=1,2)$ are roots of the following equation
\begin{equation}\label{eq6}
\begin{aligned}
&s^{2}-((-\nu_{1} \lambda +a)+(-\nu_{2} \lambda +d))s\\
&+((-\nu_{1} \lambda +a)(-\nu_{2} \lambda +d)-bc)=0,
\end{aligned}
\end{equation}
and $k_i(\lambda)(i=1,2)$ are the roots of the equation
\begin{equation}\label{eq7}
    ck^{2}+((\nu_{2}-\nu_{1})\lambda+(a-d))k-b=0.
\end{equation}
By simple calculations, the roots of (\ref{eq6}) are 
\begin{equation}\nonumber
\begin{aligned}
    &s_{1}(\lambda)=\frac{((-\nu_{1} \lambda +a)+(-\nu_{2} \lambda +d))+ \sqrt{\Delta}}{2},\\
    &s_{2}(\lambda)=\frac{((-\nu_{1} \lambda +a)+(-\nu_{2} \lambda +d))- \sqrt{\Delta}}{2},
\end{aligned}
\end{equation}
and the roots of (\ref{eq7}) are 
\begin{equation}\nonumber
\begin{aligned}
    &k_{1}(\lambda)=\frac{-((\nu_{2}-\nu_{1})\lambda+(a-d))+ \sqrt{\Delta}}{2c},\\
    &k_{2}(\lambda)=\frac{-((\nu_{2}-\nu_{1})\lambda+(a-d))- \sqrt{\Delta}}{2c},
\end{aligned}
\end{equation}
where 
$\nonumber
    \Delta=\Delta(\lambda)=((\nu_{2}-\nu_{1})\lambda+(a-d))^{2}+4bc.
$

Denote $S=S_{1}\cup S_{2}$ with
\begin{equation}\label{eq25}
\begin{aligned}
      &S_{1}=\Big \{(a,b,c,d,\nu_{1},\nu_{2})\in \mathbf{R}^{6}| there\  exist\  j,l \in \{1,2\},\\
      &\lambda_{1},\lambda_{2}\in \sigma(L), \lambda_{1}\neq \lambda_{2}, such \ that\ s_{j}(\lambda_{1})=s_{l}(\lambda_{2})\Big \}, 
\end{aligned}
\end{equation}
and 
\begin{equation}\label{eq26}
\begin{aligned}
    S_{2}=\Big \{&(a,b,c,d,\nu_{1},\nu_{2})\in \mathbf{R}^{6}| there\  exists\  \lambda \in \sigma(L), \\
    & such \ that\  \Delta(\lambda)=0 \Big \},
\end{aligned}
\end{equation}
where $\sigma(L)=\{\lambda_{i}| 1\leq i\leq N\}$ is the spectrum of $L$. It is clear that $\mathbf{m}(S)=0$, where $\mathbf{m}(\cdot)$ is the Lebesgue measure in $\mathbf{R}^{6}$.

%the usage of the set S, discuss controllability on a set of full Lebesgue measure

\begin{proposition}\label{proposition2}
    If $(a,b,c,d,\nu_{1},\nu_{2})\in \mathbf{R}^{6}-S$, $b\neq 0$, and $c\neq 0$, then $(A,\tilde{B})$ is controllable if and only if $(L,\begin{bmatrix} B & C \end{bmatrix})$ is controllable.
\end{proposition}
\begin{proof}
By Lemma \ref{lemma1}, we have necessity. Now, we consider sufficiency. By $b\neq 0$, $c\neq 0$ and $(a,b,c,d,\nu_{1},\nu_{2})\notin S_{2}$, the equation (\ref{eq7}) has two different nonzero roots $k_{j}(\lambda) (j=1,2)$, where $S_{2}$ is defined in (\ref{eq26}). The set of left eigenvectors  of $A$, i.e., $\{[v^{T},k_{j}(\lambda)v^{T}] | \lambda \in \sigma(L), v^{T}L=\lambda v^{T}, j=1,2\}$,  constructs a basis of $\mathbf{R}^{2N}$. For a given eigenvalue $s_{0}$ of $A$,  each left eigenvectors $w^{T}$  can be expressed as
\begin{equation}\label{eq22}
w^{T}=\sum_{i=1}^{N}\sum_{j=1}^{2}l_{ij}[v_{i}^{T},k_{j}(\lambda_{i})v_{i}^{T}].
\end{equation}
By $w^{T}A=s_{0} w^{T}$, we can deduce that for $1\leq i\leq N$ and $j=1,2$
\begin{equation}\label{eq21}
    l_{ij}=0,\ if\ s_{j}(\lambda_{i})\neq s_{0}.
\end{equation}
By the condition $(a,b,c,d,\nu_{1},\nu_{2})\notin S_{2}$, we have $s_{1}(\lambda_{i})\neq s_{2}(\lambda_{i})$ for $1\leq i \leq N$. Based on (\ref{eq21}), we have $l_{i1}=0$ or $l_{i2}=0$ for $1\leq i \leq N$ by reduction to absurdity. Without loss of generality, we assume that $l_{i2}=0$ for $1\leq i \leq N$, $s_{1}(\lambda_{i})=s_{0}$ for $1\leq i\leq p (1\leq p\leq N)$, and $l_{i1}=0$ for $p+1\leq i\leq N$. As $s_{1}(\lambda_{1})=\cdots=s_{1}(\lambda_{p})=s_{0}$ and $(a,b,c,d,\nu_{1},\nu_{2})\notin S_{1}$, we have $\lambda_{1}=\cdots=\lambda_{p}$ and $k_{j}(\lambda_{1})=\cdots=k_{j}(\lambda_{p}) \triangleq k_{j} (j=1,2)$, where $S_{1}$ is defined in (\ref{eq25}). Then by (\ref{eq22}) and the above assumption, we have $w^{T}=[u^{T},k_{1}u^{T}]$, where $u=\sum_{i=1}^{p}l_{i1}v_{i}$ is a left eigenvector of $L$ subject to the eigenvalue $\lambda_{1}$. Thus, every left eigenvector can be expressed as the form of $[v^{T},lv^{T}]$, where $v^{T}$ is a left eigenvector of $L$ and $l\in \mathbf{C}, l\neq 0$. 

If $(A, \tilde{B})$ is uncontrollable, then by PBH criterion, there exists a left eigenvector $[v^{T},lv^{T}]$ such that $[v^{T},lv^{T}]\tilde{B}=0$, where $v^{T}$ is a left eigenvector of $L$. Thus, we have $v^{T}\begin{bmatrix} B & C \end{bmatrix}=0$, which indicates that $(L,\begin{bmatrix} B & C \end{bmatrix})$ is uncontrollable. By reduction to absurdity, the sufficiency is proved.
\end{proof}

For both cases $B=C$ and $B\neq C$, we have proved that the controllability of $(A,\tilde{B})$ is equivalent to the controllability of $(L,\begin{bmatrix} B & C \end{bmatrix})$ except for a  set of zero Lebesgue measure.  By now, the controllability problem of a second-order multiagent system is transferred to that of a first-order multiagent system. Thus, the analysis of the controllability  problem can be simplified.

In the following section, we investigate the controllability of $(L,\bar{B})$ with $\bar{B}=\begin{bmatrix} B & C \end{bmatrix}$. 

% Thus, how the controllability of system (\ref{eq1}) is determined by the choice of control nodes, we can study how the controllability of system $\dot{z}=Lz+Bu$ is determined by the choice of control nodes, where $z\in \mathbf{R}^{N}$ and $L$, $B$, $u$ are defined in the system (\ref{eq1}).

\section{Controllability of $(L,\bar{B})$}\label{section4}

In this section, we consider controllability of the system
    \begin{equation}\label{eq8}
        \dot{z}=Lz+\bar{B}u,
    \end{equation}
where $z\in \mathbf{R}^{N}$, $u\in \mathbf{R}^{\rho}$, $L$ is the Laplacian matrix of the graph,  $\bar{B}=[e^{(N)}_{i_{1}},\cdots,e^{(N)}_{i_{\rho}}]$ is the control matrix with $1\leq i_{k}\leq N (1\leq k\leq \rho)$ and $\rho$ is the number of controlled nodes. For general graphs, it is hard to give a detailed analysis for the controllability of the system (\ref{eq8}) since eigenvalues and eigenvectors of $L$ can not be explicitly expressed. Thus, we consider the controllability on two typical graphs: grid graph and cylinder grid graph, which are shown in Fig.~\ref{fig1}.

A direct lemma derived from PBH Criterion is introduced, which gives a necessary condition for the controllability of $(L,\bar{B})$.
\begin{lemma}\label{lemma2}
  If $L$ has an eigenvalue with multiplicity $r$ and $rank(\bar{B})<r$, then $(L,\bar{B})$ is uncontrollable.
\end{lemma}

\subsection{Controllability of system (\ref{eq8}) on grid graphs}\label{subsection4.1}

In this subsection, we consider controllability of the system (\ref{eq8}) on $m\times n$ grid graphs $\mathcal{L}_{m,n}$, and the Laplacian matrix of $\mathcal{L}_{m,n}$ is denoted as $L(m,n)$. 

 We know that the grid graph $\mathcal{L}_{m,n}$ is the Cartesian product of two path graphs $\mathcal{P}_{m}$ and $\mathcal{P}_{n}$, and  we first introduce a lemma which gives the explicit expression of  eigenvalues and eigenvectors of  the Laplacian matrix $L_{\mathcal{P}}(m)$ of the path graph $\mathcal{P}_{m}$. 
 
\begin{lemma}[\cite{Yueh2005}]\label{lemma3}
  The eigenvalues $\lambda^{(\mathcal{P})}_{0},\cdots,$ $\lambda^{(\mathcal{P})}_{m-1}$ of $L_{\mathcal{P}}(m)$ are given by
\begin{equation}
  \nonumber
  \lambda^{(\mathcal{P})}_{\alpha}=2-2\cos{\frac{\alpha\pi}{m}},\  \alpha=0,1,\cdots,m-1,
\end{equation}
and the corresponding eigenvectors $v_{\alpha}=[v_{\alpha,1},\cdots,v_{\alpha,m}]^{T}$ are given by
\begin{equation}
  \label{eq50}
     v_{\alpha,k}=\begin{cases}
    \frac{1}{\sqrt{m}},&\mbox{if $\alpha = 0$;}\\
    \sqrt{\frac{2}{m}}\cos{\frac{\alpha(2k-1)\pi}{2m}},&\mbox{otherwise;}
    \end{cases}
\end{equation}
for  $k=1,2,\cdots,m$.
\end{lemma}

Based on the above lemma, we have the following results about the eigenvalues and eigenvectors of the Laplacian matrix $L(m,n)$ of the grid graph $\mathcal{L}_{m,n}$.

\begin{lemma}\label{lemma8}
    $L(m,n)$ has the following $mn$ eigenvalues,
\begin{equation}\label{eq17}
        \lambda_{\alpha \beta}= 4-2\big(\cos{\frac{\alpha\pi}{m}}+\cos{\frac{ \beta \pi}{n}}\big), 0\leq \alpha\leq m-1, 0\leq  \beta \leq n-1.
    \end{equation}
    The eigenvector subject to $\lambda_{\alpha \beta }$ is $v_{\alpha}\otimes w_{ \beta }$, where $\otimes$ denotes the Kronecker product, $v_{\alpha}=[v_{\alpha,1},\cdots,v_{\alpha,m}]^{T}$ with $v_{\alpha,k}$ defined in (\ref{eq50}) and $w_{ \beta }=[w_{ \beta ,1},\cdots,w_{ \beta ,m}]^{T}$ with
    \begin{equation}
     \label{eq51}
         w_{ \beta ,l}=\begin{cases}
        \frac{1}{\sqrt{n}},&\mbox{if $\beta = 0$;}\\
        \sqrt{\frac{2}{n}}\cos{\frac{\beta(2l-1)\pi}{2n}},&\mbox{otherwise;}
        \end{cases}
    \end{equation}     
     for $1\leq l\leq n$.
\end{lemma}

From (\ref{eq17}), we see that ``$\cos{\frac{\alpha\pi}{m}}+\cos{\frac{ \beta \pi}{n}}$" reflects the difference between $\lambda_{\alpha \beta }$ with other eigenvalues, so we call $\cos{\frac{\alpha\pi}{m}}+\cos{\frac{ \beta \pi}{n}}$ the characteristic part of  $\lambda_{\alpha \beta }$. Denote $V=[v_{0},\cdots,v_{m-1}]^{T}$ and $W=[w_{0},\cdots,w_{n-1}]^{T}$.
\begin{remark}\label{remark1}
  Since $\{v_{\alpha}\}_{0\leq \alpha\leq m-1}$ are unit eigenvectors subject to different eigenvalues of $L_{\mathcal{P}}(m)$, $V$ is an orthogonal matrix. Similarly, $W$ is an orthogonal matrix.
\end{remark}

We first consider the case of $m=n$. We denote $L(n,n)$ as the Laplacian matrix of the grid graph $\mathcal{L}_{n,n}$. Based on Lemma \ref{lemma8}, we have the following result on the largest multiplicity of the eigenvalues of $L(n,n)$.

\begin{corollary}\label{corollary2}
    For $n\geq 3$, the largest multiplicity of eigenvalues of $L(n,n)$ is $n-1$.
\end{corollary}
\begin{proof}
By Lemma \ref{lemma8}, for $ \beta =n-\alpha$ with $1\leq \alpha\leq n-1$, we have $\lambda_{\alpha \beta }= 4$. Thus, $L(n,n)$ has an eigenvalue with multiplicity  $n-1$. In the following, we show that $L(n,n)$ has no eigenvalues whose multiplicities are strictly larger than $n-1$ by reduction to absurdity. If $L(n,n)$ has an eigenvalue whose multiplicity is larger than $n-1$, then there exist two sequences $\{\alpha_{i}, 0\leq \alpha_i\leq n-1\}_{0\leq i\leq n-1}$ and $\{\beta_{i}, 0\leq \beta_i\leq n-1\}_{0\leq i\leq n-1}$ such that
\begin{equation}\label{eq23}
    \lambda_{\alpha_{0},\beta_{0}}=\lambda_{\alpha_{1},\beta_{1}}=\cdots=\lambda_{\alpha_{n-1},\beta_{n-1}},
\end{equation}
where $\alpha_{i}\neq \alpha_{j}$ and $\beta_{i}\neq \beta_{j}$ for any $0\leq i < j\leq n-1$.
%Thus, by (\ref{eq17}) and (\ref{eq23}), %we have
% \begin{equation}
%     \begin{aligned}\label{eq10}
%         \cos{\frac{\alpha_{0}\pi}{n}}+\cos{\frac{\beta_{0}\pi}{n}}=\cos{\frac{\alpha_{1}\pi}{n}}+\cos{\frac{\beta_{1}\pi}{n}}\\
%     =\cdots=\cos{\frac{\alpha_{n-1}\pi}{n}}+\cos{\frac{\beta_{n-1}\pi}{n}}.
%     \end{aligned}
% \end{equation}

Without loss of generality, we assume that $\alpha_{0}<\alpha_{1}<\cdots<\alpha_{n-1}$. Then by (\ref{eq17}) and (\ref{eq23}), we have $\beta_{0}>\beta_{1}>\cdots >\beta_{n-1}$. Thus, for $0\leq i\leq n-1$, we have $\alpha_{i}=n-1-\beta_{i}=i$. By the first equation in (\ref{eq23}), we have $1+\cos{\frac{2\pi}{n}}=2\cos{\frac{\pi}{n}}$, which contradicts with the fact $1+\cos{\frac{2\pi}{n}}=2\cos^{2}{\frac{\pi}{n}}<2\cos{\frac{\pi}{n}}$. This completes the proof of Corollary \ref{corollary2}. 
\end{proof}

We next introduce a lemma about orthogonal matrices, which will be used in the proof of the following theorems.
\begin{lemma}[\cite{Horn2012}]\label{lemma4}
    Let $W=(w_{ij})_{n\times n}$ be an orthogonal matrix. Denote
    \begin{equation}\label{eq24}
        W^{*}_{ij} = \begin{bmatrix}
        w_{11}&\cdots & w_{1,j-1}& w_{1,j+1}&\cdots & w_{1,n}\\
        \vdots&\ddots & \vdots& \vdots&\ddots & \vdots\\
        w_{i-1,1}&\cdots & w_{i-1,j-1}& w_{i-1,j+1}&\cdots & w_{i-1,n}\\
        w_{i+1,1}&\cdots & w_{i+1,j-1}& w_{i+1,j+1}&\cdots & w_{i+1,n}\\
        \vdots&\ddots & \vdots& \vdots&\ddots & \vdots\\
        w_{n,1}&\cdots & w_{n,j-1}& w_{n,j+1}&\cdots & w_{n,n}
        \end{bmatrix}.
    \end{equation}
Then  for $1\leq i,j\leq n$, $det(W^{*}_{ij})=0$ if and only if $w_{ij}=0$.
\end{lemma}

\begin{figure}
\begin{center}
\includegraphics[width=4.5cm]{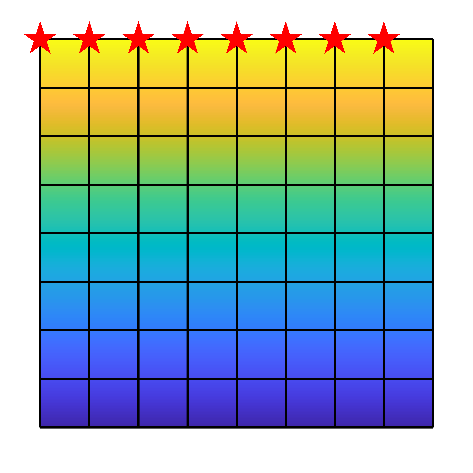} 
  \caption{This figure illustrates  how to choose control nodes on the grid graph $\mathcal{L}_{m,n}$ such that $(L(m,n),\bar{B})$ is controllable for $m=n$, where the control nodes are marked by red stars.}
  \label{fig4}                            % Size the figures 
\end{center}                                 % accordingly.
\end{figure}

Now, we give a theorem about the minimal number of control nodes such that $(L(n,n),\bar{B})$ is controllable. The meaning of ``minimal number" is  twofold. One is that there exist at least a set of control nodes whose cardinality is exactly the given minimal number such that $(L(n,n),\bar{B})$ is controllable. Another is that if the number of control nodes is less than the given minimal number, then $(L(n,n),\bar{B})$ is uncontrollable no matter how we choose control nodes.

\begin{theorem}\label{theorem1}
    Denote $\varphi (n,n)$ as the minimal number of control nodes such that $(L(n,n),\bar{B})$ is controllable, that is,
        \begin{equation}\nonumber
        \begin{aligned}
            &\varphi (n,n)= \min \big \{\rho\in \mathbf{Z}^{+}| there\  exist\  i_{1},...,i_{\rho},\ 1\leq i_{1},...,i_{\rho}\\
            &\leq n^2,\ such \ that\ (L(n,n),[e^{(n^{2})}_{i_{1}},\cdots,e^{(n^{2})}_{i_{\rho}}])\  is\\  &controllable \big \}.
        \end{aligned}
    \end{equation}
    Then we have 
    \begin{equation}\nonumber
            \varphi (n,n)=\begin{cases}
            n,&\mbox{for $n=1,2$;}\\
            $n$-1,&\mbox{for $n\geq 3$.}
            \end{cases}
            \end{equation}
\end{theorem}
\begin{proof}
For $n=1,2$, it is clear that $\varphi (n,n) = n$ by directly calculating the rank of controllability matrices.

For $n\geq 3$, by Lemma \ref{lemma2} and Corollary \ref{corollary2}, we have $\varphi (n,n)\geq n-1$. In the following, we prove $\varphi (n,n)\leq n-1$. 
We choose  the first $n-1$ nodes in the first row as control nodes as illustrated in Fig. \ref{fig4}. It is clear  that the corresponding control matrix is $$\bar{B}=e^{(n)}_{1}\otimes[e^{(n)}_{1},\cdots,e^{(n)}_{n-1}].$$ 
 By Lemma \ref{lemma8}, for any eigenvalue of $L(n,n)$ with multiplicity $r$, there exist two sequences  $\{\alpha_{i}, 0\leq \alpha_i\leq n-1\}_{1\leq i\leq r}$ and $\{\beta_{i}, 0\leq \beta_i\leq n-1\}_{1\leq i\leq r}$ such that
\begin{equation}\nonumber
    \lambda_{\alpha_{1},\beta_{1}}=\lambda_{\alpha_{2},\beta_{2}}=\cdots=\lambda_{\alpha_{r},\beta_{r}},
\end{equation}
where  $\lambda_{\alpha,\beta}$ is defined in Lemma \ref{lemma8}, and  $\alpha_{i}\neq \alpha_{j}$ and $\beta_{i}\neq \beta_{j}$ for any $1\leq i<j\leq r$. By Corollary \ref{corollary2}, we have $r\leq n-1$. Denote any one of the eigenvectors corresponding to the above eigenvalue as $u$, then there  exists a vector $g=[g_1,\cdots,g_r]^T\in \mathbf{R}^{r}\backslash\{0\}$, such that $u=\sum_{i=1}^{r}g_{i}(v_{\alpha_{i}}\otimes w_{\beta_{i}})$, where $v_{\alpha_i}$ and $w_{\beta_i}$ are defined in Lemma \ref{lemma8}. By PBH criterion, we need to show  $u^{T}\bar{B}\neq 0$ for any $g\neq 0$. 

By Lemma \ref{lemma4}, Remark \ref{remark1} and the fact that $w_{\beta,n}\neq 0$ for any $0\leq \beta \leq n-1$, we have $rank(W^{*}_{\beta,n})=n-1$, where $w_{\beta,n}$ is defined in (\ref{eq51}) and $W^{*}_{\beta,n}$ is defined in (\ref{eq24}). As $r\leq n-1$, we have
\begin{equation}
    \nonumber
    rank\left(\begin{bmatrix} w_{\beta_{1},1}&\cdots & w_{\beta_{1},n-1}\\\vdots& \ddots& \vdots \\  w_{\beta_{r},1}&\cdots & w_{\beta_{r},n-1}\end{bmatrix}\right) = r.
\end{equation}
Thus, we have
\begin{equation}\nonumber
\begin{aligned}
    &rank\left(\begin{bmatrix} (v_{\alpha_{1}}\otimes w_{\beta_{1}})^{T}\bar{B} \\ \vdots \\ (v_{\alpha_{r}}\otimes w_{\beta_{r}})^{T}\bar{B}\end{bmatrix}\right) \\
    &= rank\left(\begin{bmatrix} w_{\beta_{1},1}&\cdots & w_{\beta_{1},n-1}\\\vdots& \ddots& \vdots \\  w_{\beta_{r},1}&\cdots & w_{\beta_{r},n-1}\end{bmatrix}\right) = r.
\end{aligned}
\end{equation}
This indicates $u^{T}\bar{B}=\sum_{i=1}^{r}g_{i}(v_{\alpha_{i}}\otimes w_{\beta_{i}})^{T}\bar{B}\neq 0$ for any $g\neq 0$. Thus, system (\ref{eq8}) is controllable by choosing control nodes in the above manner. We complete the proof of the theorem. 
\end{proof}

In the following, we  consider the case of $m\neq n$.  Firstly, we introduce a lemma which gives the complete solution of the following trigonometric diophantine equation
\begin{equation}
    \label{eq12}
    \cos{(f_{1} \pi)}+\cos{(f_{2} \pi)}+\cos{(f_{3} \pi)}+\cos{(f_{4} \pi)}=0,
\end{equation}
in which all the variables $f_{j}(0\leq f_{j}\leq 1,1\leq j\leq 4)$ are rational. In the following lemma, we express a solution of (\ref{eq12}) by an unordered quadruple $\{ f_{1},f_{2},f_{3},f_{4}\}$.

\begin{lemma}[\cite{Wlodarski1969}]\label{lemma5}
    If rational numbers $f_{j}(0\leq f_{j}\leq 1,1\leq j\leq 4)$ satisfy the trigonometric diophantine equation (\ref{eq12}), then $\{ f_{1},f_{2},f_{3},f_{4}\}$ either belongs to the following infinite family (\ref{eq47}) or (\ref{eq48})
      \begin{align}
          & \left \{\gamma,\delta,\pi-\gamma,\pi-\delta \right \}\ \left(0\leq \gamma \leq \delta \leq \frac{\pi}{2} \right),\label{eq47}\\
          & \left \{\theta,\frac{2\pi}{3}-\theta,\frac{2\pi}{3}+\theta,\frac{\pi}{2} \right \}\ \left(0< \theta < \frac{\pi}{3} \right) \label{eq48}
      \end{align}
    or is one of the following quadruples
      \begin{equation}
          \label{eq49}
          \begin{aligned}
          & \left \{\frac{\pi}{3},\pi, \frac{\pi}{2},\frac{\pi}{3} \right \},\ \left \{\frac{2\pi}{3},0,\frac{\pi}{2},\frac{2\pi}{3} \right \}, \\
          & \left \{\frac{2\pi}{5},\frac{4\pi}{5}, \frac{\pi}{2},\frac{\pi}{3} \right \},\ \left \{\frac{3\pi}{5},\frac{\pi}{5},\frac{\pi}{2},\frac{2\pi}{3} \right \}, \\
          & \left \{\frac{\pi}{5},\frac{3\pi}{5},\pi,\frac{\pi}{3} \right \},\ \left \{\frac{4\pi}{5},\frac{2\pi}{5},0,\frac{2\pi}{3} \right \}, \\
          & \left \{\frac{2\pi}{5},\frac{7\pi}{15},\frac{13\pi}{15},\frac{\pi}{3} \right \},\ \left \{\frac{3\pi}{5},\frac{8\pi}{15},\frac{2\pi}{15},\frac{2\pi}{3} \right \}, \\
          & \left \{\frac{\pi}{15},\frac{4\pi}{5},\frac{11\pi}{15},\frac{\pi}{3} \right \},\ \left \{\frac{14\pi}{15},\frac{\pi}{5},\frac{4\pi}{15},\frac{2\pi}{3} \right \}, \\
          & \left \{\frac{2\pi}{7},\frac{4\pi}{7},\frac{6\pi}{7},\frac{\pi}{3} \right \},\ \left \{\frac{5\pi}{7},\frac{3\pi}{7},\frac{\pi}{7},\frac{2\pi}{3} \right \}. \\
          \end{aligned}
      \end{equation}
\end{lemma}

Denote $L(m,n)$ as the Laplacian matrix of the grid graph $\mathcal{L}_{m,n}$. By  Lemma \ref{lemma5}, we can obtain the largest multiplicity of eigenvalues of $L(m,n)$ for general $m$ and $n$.

\begin{corollary}\label{corollary3}
    Let $\psi(m,n)$ denote the largest multiplicity of eigenvalues of $L(m,n)$, and $\mathbf{d}$ denote the greatest common divisor of $m$ and $n$. Then we  have  
    \begin{itemize}
        \item for $\mathbf{d}\geq 4$, $\psi(m,n)=\mathbf{d}-1$;
        \item for $\mathbf{d}=3$,
            \begin{equation}\nonumber
            \psi(m,n)=\begin{cases}
            3,&\mbox{if $2|m$ or $2|n$;}\\
            2,&\mbox{otherwise;}
            \end{cases}
            \end{equation}
        \item for $\mathbf{d}=2$, $\psi(m,n)=2$;
        \item for $\mathbf{d}=1$, 
        % $\psi(m,n)=2$ if $2|m$, $3|n$ or $3|m$, $2|n$ or $3|m$, $5|n$ or $5|m$, $3|n$; otherwise, $\psi(m,n)=1$,
           \begin{equation}\nonumber
            \begin{aligned}
            &\psi(m,n)=\begin{cases}
            2,&\mbox{if $(2|m,3|n)$ or $(3|m,2|n)$}\\
              &\mbox{or $(3|m,5|n)$ or $(5|m,3|n)$;}\\
            1,&\mbox{otherwise;}
            \end{cases}
            \end{aligned}
            \end{equation}
    \end{itemize}
    where $p|q$ denotes that $p$ divides $q$ exactly, and $(2|m,3|n)$ means that these two conditions $2|m$ and $3|n$ are satisfied simultaneously.
\end{corollary}
\begin{proof}
By Lemma \ref{lemma8}, we need to solve the following trigonometric diophantine equation in order to find all the multiple eigenvalues of $L(m,n)$,
\begin{equation}
    \label{eq13}
    \cos{\frac{\alpha_{1}\pi}{m}}+\cos{\frac{\beta_{1}\pi}{n}}=\cos{\frac{\alpha_{2}\pi}{m}}+\cos{\frac{\beta_{2}\pi}{n}},
\end{equation}
where $\alpha_i(i=1,2)\in \{0,1,\cdots,m-1\}$ and $\beta_i(i=1,2)\in \{0,1,\cdots,n-1\}$. The solutions of (\ref{eq13}) are given by (\ref{eq47})(\ref{eq48}) and (\ref{eq49}) in Lemma \ref{lemma5}.

We further discuss the solutions given by (\ref{eq47}), that is,
\begin{equation}
    \label{eq52}
    \left \{\frac{\alpha_{1}\pi}{m},\frac{\beta_{1}\pi}{n},\pi-\frac{\alpha_{2}\pi}{m},\pi-\frac{\beta_{2}\pi}{n} \right \} = \left \{\gamma,\delta,\pi-\gamma,\pi-\delta \right \},
\end{equation}
where $0\leq \gamma \leq \delta \leq \frac{\pi}{2}$. The equation (\ref{eq52}) indicates that in the unordered quadruple $\left \{\frac{\alpha_{1}\pi}{m},\frac{\beta_{1}\pi}{n},\pi-\frac{\alpha_{2}\pi}{m},\pi-\frac{\beta_{2}\pi}{n} \right \}$, there exist two terms whose sum is $\pi$ and sum of the other two terms is also $\pi$. If $\frac{\alpha_{1}\pi}{m}+\frac{\beta_{1}\pi}{n}=\pi$, then the solutions of (\ref{eq13}) are
\begin{equation}\label{eq27}
    \alpha_{i}=\frac{k_{i}m}{\mathbf{d}},\ \beta_{i}=\frac{(\mathbf{d}-k_{i})n}{\mathbf{d}}
\end{equation}
for $i=1,2$ and $1\leq k_{i}\leq \mathbf{d}-1$. If $\frac{\alpha_{1}\pi}{m}+(\pi-\frac{\alpha_{2}\pi}{m})=\pi,$ then we have $\alpha_1=\alpha_2$ and $\beta_1=\beta_2$. If $\frac{\alpha_{1}\pi}{m}+(\pi-\frac{\beta_{2}\pi}{n})=\pi,$ then the solutions of (\ref{eq13}) are
\begin{equation}\nonumber
    \alpha_{i}=\frac{k_{i}m}{\mathbf{d}},\ \beta_{3-i}=\frac{k_{i}n}{\mathbf{d}}
\end{equation}
for $i=1,2$ and $1\leq k_{i}\leq \mathbf{d}-1$. Solutions of (\ref{eq13}) given in (\ref{eq48}) and (\ref{eq49}) can be obtained similarly. By these solutions, we can obtain all the multiple eigenvalues and corresponding multiplicities.

 We see that the eigenvalue $4$ of $L(m,n)$ is of $\mathbf{d}-1$ multiplicities by (\ref{eq27}). Thus, 
 \begin{equation}
     \label{eq29}
     \psi(m,n)\geq \mathbf{d}-1.
 \end{equation}

\begin{itemize}
\item For $\mathbf{d}\geq 4$, if $L(m,n)$ has an eigenvalue with multiplicity $r(r\geq 4)$, then according to the basic property of exact division and the explicit expression of solutions of (\ref{eq13}) given in Lemma \ref{lemma5}, we can easily deduce that
\begin{equation}
    \label{eq30}
    \mathbf{d}-1 \geq r
\end{equation}
by checking all the eigenvalues whose multiplicity is not less than 4. Now, we prove $\psi(m,n)= \mathbf{d}-1$ by reduction to absurdity. Assume that there exist $m$ and $n$ such that $\psi(m,n)> \mathbf{d}-1\geq 3$, then by (\ref{eq30}) we have $\mathbf{d}-1\geq\psi(m,n)$. This is a contradiction. Thus, for $\mathbf{d}\geq 4$, $\psi(m,n)= \mathbf{d}-1$. 
\item For $\mathbf{d}=3$, by the explicit expression of solutions of (\ref{eq13}) given in Lemma \ref{lemma5}, we see that if $2|m$ (or $2|n$), then the only eigenvalue of $L(m,n)$ whose multiplicity is strictly larger than $\mathbf{d}-1$ has the characteristic parts
$
        \cos{0\pi}+\cos{\frac{2\pi}{3}}=\cos{\frac{2\pi}{3}}+\cos{0\pi}=\cos{\frac{\pi}{2}}+\cos{\frac{\pi}{3}}
$(or its equivalent form). Otherwise, multiplicities of all the eigenvalues of $L(m,n)$ are not more than $\mathbf{d}-1$. Thus, for $\mathbf{d}=3$, $\psi(m,n)=3$ if $2|m$ or $2|n$; otherwise, $\psi(m,n)=\mathbf{d}-1=2$. 
\item For $\mathbf{d}=2$, $L(m,n)$ must have an eigenvalue with multiplicity $2$, which has the following characteristic part 
$
        \cos{\frac{\pi}{2}}+\cos{0\pi}=\cos{0\pi}+\cos{\frac{\pi}{2}},
$
and have no eigenvalues whose multiplicities are strictly larger than 2. Thus, for $\mathbf{d}=2$, $\psi(m,n)=2$. 
\item For $\mathbf{d}=1$, by the expression of solutions of (\ref{eq13}) given in Lemma \ref{lemma5}, we know that $L(m,n)$ has no eigenvalues whose multiplicities are strictly larger than 2. Furthermore, if ($2|m$, $3|n$) or ($3|m$, $2|n$) or ($3|m$, $5|n$) or ($5|m$, $3|n$), the eigenvalues with multiplicity 2 have at least one of the following characteristic parts or their equivalent forms
\begin{equation}
    \nonumber
    \begin{aligned}
        & \cos{\frac{\pi}{2}}+\cos{\frac{\pi}{3}}=\cos{0\pi}+\cos{\frac{2\pi}{3}},\\
        & \cos{0\pi}+\cos{\frac{2\pi}{5}}=\cos{\frac{\pi}{3}}+\cos{\frac{\pi}{5}},\\
        & \cos{\frac{\pi}{2}}+\cos{\frac{\pi}{5}}=\cos{\frac{\pi}{3}}+\cos{\frac{2\pi}{5}}.
    \end{aligned}
\end{equation}
%Here, some equations which can be derived from the above equations are not listed, such as
%\begin{equation}
 %   \cos{\frac{\pi}{2}}+\cos{\frac{3\pi}{5}}=\cos{\frac{\pi}{3}}+\cos{\frac{4\pi}{5}}.
%\end{equation}
Thus, $\psi(m,n)=2$. Otherwise, $\psi(m,n)=1$.
\end{itemize}
This completes the proof. 
\end{proof}

Based on the above corollary, we establish our theorem about the minimal number of control nodes such that $(L(m,n),\bar{B})$ is controllable when $m\neq n$.

\begin{theorem}\label{theorem2}
    Denote $\varphi (m,n)$ as the minimal number of control nodes such that $(L(m,n),\bar{B})$ is controllable, that is,
    \begin{equation}\nonumber
        \begin{aligned}
            &\varphi (m,n)= \min \big \{\rho\in \mathbf{Z}^{+}| there\  exist\  i_{1},...,i_{\rho}, 1\leq \\ 
            &i_{1},...,i_{\rho}\leq mn,\ such \ that\ (L(m,n),[e^{(mn)}_{i_{1}},\cdots,e^{(mn)}_{i_{\rho}}])\\  
            &is\ controllable \big \}.
        \end{aligned}
    \end{equation}
    Then we have $\varphi (m,n) = \psi(m,n)$, where $\psi(m,n)$ is defined in Corollary \ref{corollary3}.
\end{theorem}
\begin{proof}
By Lemma \ref{lemma2}, we see that $\varphi (m,n) \geq \psi(m,n)$. In the following, we prove $\varphi (m,n) \leq \psi(m,n)$.
%\begin{bmatrix}
 %   g_1& g_2&\cdots &g_r
By Corollary \ref{corollary3}, it is clear that  $\psi(m,n)\leq \min(m,n)$.  As illustrated in Fig.~\ref{fig4_1}, we choose the first $\psi(m,n)$ nodes in the first row as control nodes from left to right. The corresponding control matrix is $$\bar{B}=e^{(m)}_{1}\otimes[e^{(n)}_{1},\cdots,e^{(n)}_{\psi(m,n)}].$$ We next prove that $(L(m,n),\bar{B})$ is controllable. Similar to the proof of Theorem \ref{theorem1}, for any eigenvalue $\lambda$ of $L(m,n)$  and the corresponding eigenvector $u$, we just need to prove 
\begin{equation}
    \label{eq34}
    u^{T}\bar{B}\neq 0.
\end{equation}

By the explicit expression of eigenvectors in Lemma \ref{lemma8}, it is clear that (\ref{eq34}) holds for any eigenvectors $u$ corresponding to single eigenvalues. In the following, we will discuss (\ref{eq34}) for multiple eigenvalues by using  Lemma \ref{lemma8} and Corollary \ref{corollary3}. 

We introduce a matrix $Y=[y_{i,j}]_{\mathbf{d}\times \mathbf{d}}$ whose rows are $\mathbf{d}$ eigenvectors of the Laplacian matrix of $\mathcal{P}_{\mathbf{d}}$, where $\mathbf{d}$ is the greatest common divisor of $m$ and $n$ and
\begin{equation}
  \nonumber
     y_{i,j}=\begin{cases}
    \frac{1}{\sqrt{\mathbf{d}}},&\mbox{if $i = 1$;}\\
    \sqrt{\frac{2}{\mathbf{d}}}\cos{\frac{(i-1)(2j-1)\pi}{2\mathbf{d}}},&\mbox{otherwise.}
    \end{cases}
\end{equation}
Obviously, $Y$ is an orthogonal matrix. By Lemma \ref{lemma4} and the fact that $y_{i,\mathbf{d}}\neq 0$ for any $1\leq i\leq \mathbf{d}$, we have 
\begin{equation}
    \label{eq31}
    rank(Y^{*}_{i,\mathbf{d}})=\mathbf{d}-1,
\end{equation}
where $Y^{*}_{i,\mathbf{d}}$ is defined in (\ref{eq24}).

For eigenvalues with the following characteristic parts
\begin{equation}
    \nonumber
    \cos{\frac{\alpha_{1}\pi}{m}}+\cos{\frac{\beta_{1}\pi}{n}}=\cdots=\cos{\frac{\alpha_{r}\pi}{m}}+\cos{\frac{\beta_{r}\pi}{n}},
\end{equation}
where $r (2\leq r\leq \mathbf{d}-1)$ is the multiplicity of eigenvalues and $l_i\triangleq\frac{\mathbf{d}}{n} \beta_i $ are integers, we then have by (\ref{eq29}) and (\ref{eq31}) 
\begin{align}
    &rank\begin{bmatrix} (v_{\alpha_{1}}\otimes w_{\beta_{1}})^{T}\bar{B} \\ \vdots \\ (v_{\alpha_{r}}\otimes w_{\beta_{r}})^{T}\bar{B}\end{bmatrix} \nonumber \\
    \geq \  &rank \begin{bmatrix} y_{l_{1}+1,1}&\cdots & y_{l_{1}+1,d-1}\\\vdots& \ddots& \vdots \\  y_{l_{r}+1,1}&\cdots & y_{l_{r}+1,\mathbf{d}-1}\end{bmatrix} \nonumber \\
    =\ & r, \label{eq33}
\end{align}

where $\{v_{\alpha_{i}}\otimes w_{\beta_{i}}, 1\leq i\leq r\}$ is a basis for the eigenvector space corresponding to the above eigenvalue. Thus, the eigenvectors $u$ can be expressed into the following form, \begin{equation}
    \label{eq32}
    u=\sum_{i=1}^{r}g_{i}(v_{\alpha_{i}}\otimes w_{\beta_{i}}),
\end{equation} where $[g_{1}, g_{2}, \cdots, g_r]^{T}$ is a nonzero vector. 
For eigenvectors  (\ref{eq32}),  we have by (\ref{eq33}),
$$u^{T}\bar{B}=\sum_{i=1}^{r}g_{i}(v_{\alpha_{i}}\otimes w_{\beta_{i}})^{T}\bar{B}\neq 0.$$

For other multiple eigenvalues, we can list them one by one according to the explicit expression of solutions of (\ref{eq13}) given in Lemma \ref{lemma5}. We see that the mulitiplicities of these eigenvalues are not more than $5$. By directly calculating the determinant of the following matrix whose dimension is not more than $5$,
\begin{equation}
    \label{eq15}
    H = \begin{bmatrix} w_{\beta_{1},1}&\cdots & w_{\beta_{1},r}\\\vdots& \ddots& \vdots \\  w_{\beta_{r},1}&\cdots & w_{\beta_{r},r}\end{bmatrix},
\end{equation}
where $r (2\leq r\leq 5)$ is multiplicity of the eigenvalue and $w_{\beta,l}$ is defined in Lemma \ref{lemma8}, we have $det (H)\neq 0$. Thus, we can obtain (\ref{eq33}), which indicates that $u^{T}\bar{B}\neq 0$. This completes the proof of the theorem. 
\end{proof}
%For example, some of these eigenvalues have the following characteristic part,
%\begin{equation}
 %   \nonumber
  %  \begin{aligned}
   %     &\cos{\frac{l\pi}{k}}+\cos{\frac{\pi}{2}}
%         =\cos{\frac{(k-3l)\pi}{3k}} +\cos{\frac{(k+3l)\pi}{3k}}\\
%        &=\cos{\frac{(k+3l)\pi}{3k}}+\cos{\frac{(k-3l)\pi}{3k}},
%    \end{aligned}
%\end{equation}
%where $0<\frac{l}{k}<\frac{1}{6}$. For this eigenvalue, since
% \begin{equation}
% \begin{aligned}
%     \nonumber
%     &det\left(\begin{bmatrix} \cos{\frac{\pi}{2\times2}}&\cos{\frac{3\pi}{2\times2}}&\cos{\frac{5\pi}{2\times2}}\\
%     \cos{\frac{(k+3l)\pi}{2\times3k}}& \cos{\frac{3(k+3l)\pi}{2\times3k}}&\cos{\frac{5(k+3l)\pi}{2\times3k}}\\ \cos{\frac{(k-3l)\pi}{2\times3k}}& \cos{\frac{3(k-3l)\pi}{2\times3k}}&\cos{\frac{5(k-3l)\pi}{2\times3k}}
%     \end{bmatrix}\right)\\
%     &=\sqrt{6}\cos{\frac{(k+3l)\pi}{3k}}\cos{\frac{(k-3l)\pi}{3k}}(\sin{\frac{2l\pi}{k}}+\sin{\frac{l\pi}{k}})\neq 0,
% \end{aligned}
% \end{equation}
% for $0< \frac{l}{k}<\frac{1}{6}$, 
% we have $u^{T}\bar{B}\neq 0$. 

For given $m$ and $n$,  Theorem \ref{theorem2} provides the minimal number of control nodes to guarantee the controllability of  $(L(m,n),\bar{B})$. The results in Theorem \ref{theorem1} can  also be contained in Theorem \ref{theorem2}.

\begin{figure}
\begin{center}
\includegraphics[width=5.0cm]{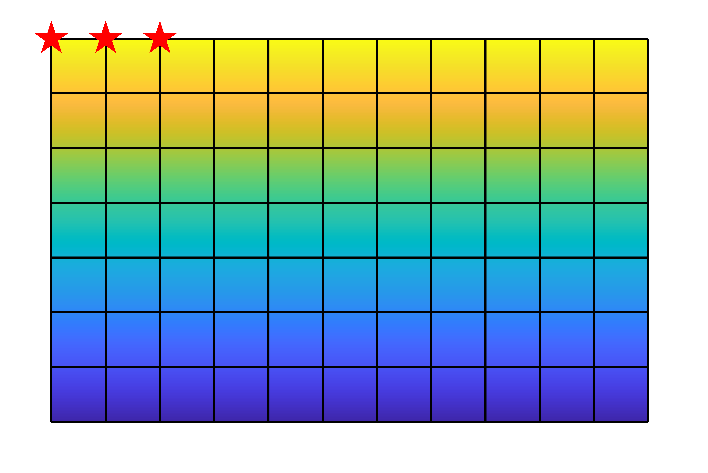} 
  \caption{This figure illustrates  how to choose control nodes on the grid graph $\mathcal{L}_{m,n}$ such that $(L(m,n),\bar{B})$ is controllable for $m\neq n$, where the control nodes are marked by red stars.}
  \label{fig4_1}                            % Size the figures 
\end{center}                                 % accordingly.
\end{figure}

Generally speaking, it is difficult to give all the control node sets such that  $(L(m,n),\bar{B})$ is controllable for  given $m$ and $n$. However, for some special cases, we can solve this problem completely by using our results. 

\begin{itemize}
\item For the case $\psi(m,n)=1$, Notarstefano and Parlangeli  in \cite{Notarstefano2013} gave all and only the combinations of control nodes for $(L(m,n),\bar{B})$ to be controllable, though they did not show under what conditions $\psi(m,n)=1$ holds. Our Corollary \ref{corollary3} gives the answer for this question. Thus, for the case $\psi(m,n)=1$, the problem of finding all the control node sets for the pair $(L(m,n),\bar{B})$ to be controllable is thoroughly solved.

\item Another case where this problem can be solved thoroughly is $\mathbf{d}=1$, $2|m$, $3|n$ and $5\nmid m$ so that $\varphi(m,n)=\psi(m,n)=2$ (see Corollary \ref{corollary3} and Theorem \ref{theorem2}), where $5\nmid m$ means that $5$ is not a factor of $m$. For this case, $L(m,n)$ has only one multiple eigenvalue whose multiplicity is 2. We can find all the control node combinations such that (\ref{eq33}) holds by calculating the determinant of $2\times 2$ square matrices. Thus, all the control node sets such that $(L(m,n),\bar{B})$ is controllable can be obtained.
\end{itemize} 

\subsection{Controllability of system (\ref{eq8}) on cylinder grid graphs}

In Subsection \ref{subsection4.1}, we investigate controllability of  system (\ref{eq8}) on grid graphs. In fact, the methods in the above subsection can also be extended to study controllability of system (\ref{eq8}) on $m\times n$ cylinder grid graphs $\mathcal{C}_{m,n}$. The Laplacian matrix of $\mathcal{C}_{m,n}$ is denoted as $L_{\mathcal{C}}(m,n)$.

By Subsection \ref{subsection2.1} , we know that a cylinder grid graph $\mathcal{C}_{m,n}$ is the Cartesian product of a cycle graph $\mathcal{D}_{m}$ and a path graph $\mathcal{P}_{n}$. In order to provide properties of eigenvalues and eigenvectors of $L_{\mathcal{C}}(m,n)$, we introduce a lemma which gives the explicit expression of eigenvalues and eigenvectors of the Laplacian matrix $L_{\mathcal{D}}(m)$ of the cycle graph $\mathcal{D}_{m}$. 

\begin{lemma}[\cite{Godsil2001}]\label{lemma6}
  The eigenvalues $\lambda^{(\mathcal{D})}_{1},\cdots,\lambda^{(\mathcal{D})}_{m}$ of $L_{\mathcal{D}}(m)$ are given by
\begin{equation}
  \label{eq36}
  \lambda^{(\mathcal{D})}_{\alpha}=2-2\cos{\frac{2 \alpha \pi}{m}},\   \alpha =1,2,\cdots,m,
\end{equation}
and the corresponding eigenvectors $\bar{v}_{ \alpha }=[\bar{v}_{ \alpha ,1},\cdots,\bar{v}_{ \alpha ,m}]^{T}$ are given by
\begin{equation}
  \nonumber
  \bar{v}_{ \alpha ,j}=\mathbf{e}^{\mathbf{i}\frac{2 \alpha j\pi}{m}},\  j =1,2,\cdots,m,
\end{equation}
where $\mathbf{i}$ is the unit of imaginary number and $\mathbf{e}$ is the base of natural logarithm.
\end{lemma}

Based on the above lemma, we have the following results about eigenvalues and eigenvectors of the Laplacian matrix $L_{\mathcal{C}}(m,n)$.

\begin{lemma}\label{lemma9}
   The Laplacian matrix $L_{\mathcal{C}}(m,n)$ has the following $mn$ eigenvalues,
    \begin{equation}\label{eq45}
         \hat{ \lambda}_{ \alpha  \beta }= 4-2(\cos{\frac{2 \alpha \pi}{m}}+\cos{\frac{ \beta \pi}{n}}), 1\leq  \alpha \leq m, 0\leq  \beta \leq n-1.
    \end{equation}
    The eigenvector corresponding to $\hat{ \lambda}_{ \alpha  \beta }$ is $\bar{v}_{ \alpha }\otimes w_{ \beta }$, where $\bar{v}_{ \alpha }=[\bar{v}_{ \alpha ,1},\cdots,\bar{v}_{ \alpha ,m}]^{T}$ with $\bar{v}_{ \alpha ,k}=\mathbf{e}^{\mathbf{i}\frac{2 \alpha k\pi}{m}}$ for $1\leq k\leq m$ and $w_{ \beta }=[w_{ \beta ,1},\cdots,w_{ \beta ,n}]^{T}$ with $w_{ \beta ,l}$ defined in (\ref{eq51}).
\end{lemma}

Parallel to Corollary \ref{corollary3}, we have the following result for the largest multiplicity of eigenvalues of $L_{\mathcal{C}}(m,n)$ by using Lemma \ref{lemma5}.

\begin{corollary}\label{corollary5}
Let $\psi_{\mathcal{C}}(m,n)(m\geq 3,n\geq 2)$ denote the largest multiplicity of  eigenvalues of $L_{\mathcal{C}}(m,n)$, and $\mathbf{d}_{\mathcal{C}}$ denote the greatest common divisor of $m$ and $2n$. We have  
    \begin{itemize}
        \item for $\mathbf{d}_{\mathcal{C}}\geq 6$, $\psi_{\mathcal{C}}(m,n)=\mathbf{d}_{\mathcal{C}}-1$;
        \item for $\mathbf{d}_{\mathcal{C}}=5$, 
        \begin{equation}\nonumber
            \psi_{\mathcal{C}}(m,n)=\begin{cases}
            6,&\mbox{if $(15|m,10|n)$;}\\
            4&\mbox{otherwise;}
            \end{cases}
            \end{equation}
        % $\psi_{\mathcal{C}}(m,n)=6$ if $15|m$, $20|2n$; otherwise, $\psi_{\mathcal{C}}(m,n)=4$;
        \item for $\mathbf{d}_{\mathcal{C}}=4$, 
         \begin{equation}\nonumber
         \begin{aligned}
           &\psi_{\mathcal{C}}(m,n)\\
           &=\begin{cases}
            4,&\mbox{if ($10|m$, $3|n$) or ($12|m$, $5|n$) or ($6|m$, $2|n$);}\\
            3&\mbox{otherwise;}
            \end{cases}
            \end{aligned}
            \end{equation}
        % $\psi_{\mathcal{C}}(m,n)=4$ if $10|m$, $6|2n$ or $12|m$, $10|2n$ or $6|m$, $4|2n$; otherwise, $\psi_{\mathcal{C}}(m,n)=3$;
        \item for $\mathbf{d}_{\mathcal{C}}=3$, $\psi_{\mathcal{C}}(m,n)=3$;
        \item for $\mathbf{d}_{\mathcal{C}}=2$, 
            \begin{equation}\nonumber
            \begin{aligned}
            &\psi_{\mathcal{C}}(m,n)\\
            &=\begin{cases}
            4,&\mbox{if $(10|m,3|n)$ or $(12|m,5|n)$ or $(6|m,2|n)$;}\\
            3,&\mbox{otherwise if $(6|m,5|n)$ or $(4|m,3|n)$;}\\
            2,&\mbox{otherwise;}
            \end{cases}
            \end{aligned}
            \end{equation}
        % $\psi_{\mathcal{C}}(m,n)=4$ if $10|m$, $6|2n$ or $12|m$, $10|2n$ or $6|m$, $4|2n$; $\psi_{\mathcal{C}}(m,n)=3$ if $6|m$, $10|2n$ or $4|m$, $6|2n$ but the above conditions are not satisfied; otherwise, $\psi_{\mathcal{C}}(m,n)=2$;
        \item for $\mathbf{d}_{\mathcal{C}}=1$, $\psi_{\mathcal{C}}(m,n)=2$.
    \end{itemize}
\end{corollary}
The proof of this corollary follows that of Corollary \ref{corollary3}. 

Based on Lemma \ref{lemma5}, Lemma \ref{lemma9} and Corollary \ref{corollary5}, we can establish the following theorem about the minimal number of control nodes such that $(L_{\mathcal{C}}(m,n),\bar{B})$ is controllable.

\begin{figure}
\begin{center}
\includegraphics[height=4.4cm]{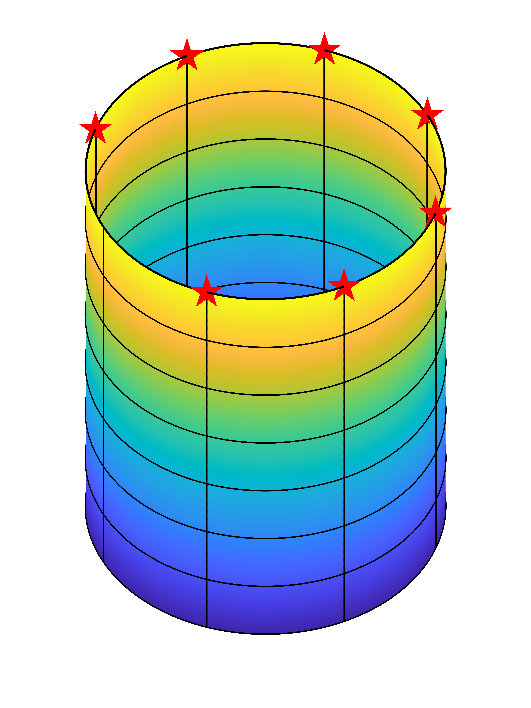} 
  \caption{This figure illustrates  how to choose control nodes on the cylinder grid graph $\mathcal{C}_{m,n}$ such that $(L_{\mathcal{C}}(m,n),\bar{B})$ is controllable, where the control nodes are marked by red stars. }
  \label{fig5}                            % Size the figures 
\end{center}                                 % accordingly.
\end{figure}

\begin{theorem}\label{theorem3}
Denote $\varphi_{\mathcal{C}} (m,n)$ as the minimal number of control nodes such that $(L_{\mathcal{C}}(m,n),\bar{B})$ is controllable, that is,
        \begin{equation}\nonumber
        \begin{aligned}
            &\varphi_{\mathcal{C}} (m,n)= \min \big\{\rho\in \mathbf{Z}^{+}| there\  exist\  i_{1},...,i_{\rho}, 1\leq \\ 
            &i_{1},...,i_{\rho}
            \leq mn, such \ that\ (L_{\mathcal{C}}(m,n),[e^{(mn)}_{i_{1}},\cdots,e^{(mn)}_{i_{\rho}}])\  \\&is\ controllable \big \}.
        \end{aligned}
    \end{equation}
    Then we have $\varphi_{\mathcal{C}} (m,n) = \psi_{\mathcal{C}}(m,n)$.
\end{theorem}
\begin{proof}
The control nodes to guarantee controllability of $(L_{\mathcal{C}}(m,n),\bar{B})$ can be chosen as $\psi_{\mathcal{C}}(m,n)$ consecutive nodes in the first row of the cylinder grid graph as illustrated in Fig.~\ref{fig5}.  It is clear that the corresponding control matrix is $\bar{B}=[e^{(m)}_{1},\cdots,e^{(m)}_{\psi_{\mathcal{C}}(m,n)}]\otimes e^{(n)}_{1}$. By following the proof line of Theorem \ref{theorem2}, we can obtain the reuslts of the theorem.
\end{proof}

\section{Concluding remarks}
This paper investigates controllability of networked multiagent systems which are obtained by linearizing and  spatial discretizing Turing's model.  We first establish the equivalence between controllability of the second-order linearized Turing's model and  the Laplace dynamic system under mild conditions on parameters of the model. Then, with the help of solutions
of the trigonometric diophantine equation, we characterize  multiplicities of eigenvalues of the Laplacian matrix of grid graphs and cylinder grid graphs. Based on this analysis, we  give a complete characterization for the controllability problem of the linearized Turing's model on these two graphs by not only providing minimal numbers of control nodes  but also choosing corresponding control node sets to guarantee controllability of the system.

\bibliographystyle{apalike}        % Include this if you use bibtex 
\bibliography{autosam}           % and a bib file to produce the 

\begin{thebibliography}{}

\bibitem[Aguilar and Gharesifard, 2015]{Aguilar2015}
Aguilar, C. and Gharesifard, B. (2015).
\newblock Graph controllability classes for the
  \uppercase\expandafter{L}aplacian leader-follower dynamics.
\newblock {\em IEEE Transactions on Automatic Control}, 60(6):1611--1623.

\bibitem[Bacciotti, 2019]{Bacciotti2019}
Bacciotti, A. (2019).
\newblock {\em Stability and Control of Linear Systems}.
\newblock Springer International Publishing, Cham.

\bibitem[Cao et~al., 2013]{Cao2013}
Cao, M., Zhang, S., and Camlibel, M. (2013).
\newblock A class of uncontrollable diffusively coupled multiagent systems with
  multichain topologies.
\newblock {\em IEEE Transactions on Automatic Control}, 58(2):465--469.

\bibitem[Chapman et~al., 2014]{Chapman2014}
Chapman, A., Nabi-Abdolyousefi, M., and Mesbahi, M. (2014).
\newblock Controllability and observability of network-of-networks via
  \uppercase\expandafter{C}artesian products.
\newblock {\em IEEE Transactions on Automatic Control}, 59(10):2668--2679.

\bibitem[Commault and Dion, 2013]{Commault2013}
Commault, C. and Dion, J.-M. (2013).
\newblock Input addition and leader selection for the controllability of
  graph-based systems.
\newblock {\em Automatica}, 49(11):3322--3328.

\bibitem[Conway and Jones, 1976]{Conway1976}
Conway, J. and Jones, A. (1976).
\newblock Trigonometric diophantine equations (on vanishing sums of roots of
  unity).
\newblock {\em Acta Arithmetica}, 30:229--240.

\bibitem[Gierer and Meinhardt, 1972]{Gierer1972}
Gierer, A. and Meinhardt, H. (1972).
\newblock A theory of biological pattern formation.
\newblock {\em Kybernetik}, 12:30--39.

\bibitem[Godsil, 2012]{Godsil2012}
Godsil, C. (2012).
\newblock Controllable subsets in graphs.
\newblock {\em Annals of Combinatorics}, 16(4):733--744.

\bibitem[Godsil and Royle, 2001]{Godsil2001}
Godsil, C. and Royle, G. (2001).
\newblock {\em Algebraic Graph Theory}.
\newblock Springer New York, New York.

\bibitem[Hao et~al., 2018]{Hao2018}
Hao, Y., Duan, Z., and Chen, G. (2018).
\newblock Further on the controllability of networked
  \uppercase\expandafter{MIMO LTI} systems.
\newblock {\em International Journal of Robust and Nonlinear Control},
  28(5):1778--1788.

\bibitem[Hao et~al., 2019a]{Hao2019b}
Hao, Y., Duan, Z., Chen, G., and Wu, F. (2019a).
\newblock New controllability conditions for networked, identical
  \uppercase\expandafter{LTI} systems.
\newblock {\em IEEE Transactions on Automatic Control}, 64(10):4223--4228.

\bibitem[Hao et~al., 2019b]{Hao2019a}
Hao, Y., Wang, Q., Duan, Z., and Chen, G. (2019b).
\newblock Controllability of \uppercase\expandafter{K}ronecker product
  networks.
\newblock {\em Automatica}, 110:108597.

\bibitem[Horn and Johnson, 2012]{Horn2012}
Horn, R. and Johnson, C. (2012).
\newblock {\em Matrix Analysis}.
\newblock Cambridge University Press, Cambridge, 2nd edition.

\bibitem[Ji et~al., 2012]{Ji2012}
Ji, Z., Lin, H., and Yu, H. (2012).
\newblock Leaders in multi-agent controllability under consensus algorithm and
  tree topology.
\newblock {\em Systems $\&$ Control Letters}, 61(9):918--925.

\bibitem[Ji and Yu, 2017]{Ji2017}
Ji, Z. and Yu, H. (2017).
\newblock A new perspective to graphical characterization of multiagent
  controllability.
\newblock {\em IEEE Transactions on Cybernetics}, 47(6):1471--1483.

\bibitem[Maini et~al., 2012]{Maini2012}
Maini, P., Woolley, T., Baker, R., Gaffney, E., and Lee, S. (2012).
\newblock Turing's model for biological pattern formation and the robustness
  problem.
\newblock {\em Interface focus}, 2(4):487--496.

\bibitem[Mousavi et~al., 2021]{Mousavi2021}
Mousavi, S., Haeri, M., and Mesbahi, M. (2021).
\newblock Laplacian dynamics on cographs: Controllability analysis through
  joins and unions.
\newblock {\em IEEE Transactions on Automatic Control}, 66(3):1383--1390.

\bibitem[Nabi-Abdolyousefi and Mesbahi, 2013]{Nabi-Abdolyousefi2013}
Nabi-Abdolyousefi, M. and Mesbahi, M. (2013).
\newblock On the controllability properties of circulant networks.
\newblock {\em IEEE Transactions on Automatic Control}, 58(12):3179--3184.

\bibitem[Notarstefano and Parlangeli, 2013]{Notarstefano2013}
Notarstefano, G. and Parlangeli, G. (2013).
\newblock Controllability and observability of grid graphs via reduction and
  symmetries.
\newblock {\em IEEE Transactions on Automatic Control}, 58(7):1719--1731.

\bibitem[Olshevsky, 2014]{Olshevsky2014}
Olshevsky, A. (2014).
\newblock Minimal controllability problems.
\newblock {\em IEEE Transactions on Control of Network Systems}, 1(3):249--258.

\bibitem[Parlangeli and Notarstefano, 2012]{Parlangeli2012}
Parlangeli, G. and Notarstefano, G. (2012).
\newblock On the reachability and observability of path and cycle graphs.
\newblock {\em IEEE Transactions on Automatic Control}, 57(3):743--748.

\bibitem[Rahmani et~al., 2009]{Rahmani2009}
Rahmani, A., Ji, M., Mesbahi, M., and Egerstedt, M. (2009).
\newblock Controllability of multi-agent systems from a graph-theoretic
  perspective.
\newblock {\em SIAM Journal on Control and Optimization}, 48(1):162--186.

\bibitem[Raspopovic et~al., 2014]{Raspopovic2014}
Raspopovic, J., Marcon, L., Russo, L., and Sharpe, J. (2014).
\newblock Digit patterning is controlled by a bmp-sox9-wnt turing network
  modulated by morphogen gradients.
\newblock {\em Science}, 345(6196):566--570.

\bibitem[She et~al., 2020]{She2020}
She, B., Mehta, S., Ton, C., and Kan, Z. (2020).
\newblock Controllability ensured leader group selection on signed multiagent
  networks.
\newblock {\em IEEE Transactions on Cybernetics}, 50(1):222--232.

\bibitem[She et~al., 2021]{She2021}
She, B., Mehta, S.~S., Doucette, E., Ton, C., and Kan, Z. (2021).
\newblock Characterizing energy-related controllability of composite complex
  networks via graph product.
\newblock {\em IEEE Transactions on Automatic Control}, 66(7):3205--3212.

\bibitem[Sheth et~al., 2012]{Sheth2012}
Sheth, R., Marcon, L., Bastida, M., Junco, M., Quintana, L., Dahn, R., Kmita,
  M., Sharpe, J., and Ros, M. (2012).
\newblock Hox genes regulate digit patterning by controlling the wavelength of
  a \uppercase\expandafter{T}uring-type mechanism.
\newblock {\em Science}, 338(6113):1476--1480.

\bibitem[Sick et~al., 2006]{Sick2006}
Sick, S., Reinker, S., Timmer, J., and Schlake, T. (2006).
\newblock \uppercase\expandafter{WNT} and \uppercase\expandafter{DKK} determine
  hair follicle spacing through a reaction-diffusion mechanism.
\newblock {\em Science}, 314(5804):1447--1450.

\bibitem[Sun et~al., 2017]{Sun2017}
Sun, C., Hu, G., and Xie, L. (2017).
\newblock Controllability of multiagent networks with antagonistic
  interactions.
\newblock {\em IEEE Transactions on Automatic Control}, 62(10):5457--5462.

\bibitem[Tanner, 2004]{Tanner2004}
Tanner, H. (2004).
\newblock On the controllability of nearest neighbor interconnections.
\newblock In {\em Proceedings of the 43rd IEEE Conference on Decision and
  Control}, 3:2467-2472, Nassau, Bahamas.

\bibitem[Trumpf and Trentelman, 2019]{Trumpf2019}
Trumpf, J. and Trentelman, H. (2019).
\newblock Controllability and stabilizability of networks of linear systems.
\newblock {\em IEEE Transactions on Automatic Control}, 64(8):3391--3398.

\bibitem[Turing, 1952]{Turing1952}
Turing, A. (1952).
\newblock The chemical basis of morphogenesis.
\newblock {\em Philosophical Transactions of the Royal Society of London.
  Series B, Biological Sciences}, 237(641):37--72.

\bibitem[Wang et~al., 2016]{Wang2016}
Wang, L., Chen, G., Wang, X., and Tang, W. (2016).
\newblock Controllability of networked \uppercase\expandafter{MIMO} systems.
\newblock {\em Automatica}, 69:405--409.

\bibitem[Wlodarski, 1969]{Wlodarski1969}
Wlodarski, L. (1969).
\newblock On the equation
  $\cos{\alpha_{1}}+\cos{\alpha_{2}}+\cos{\alpha_{3}}+\cos{\alpha_{4}}=0$.
\newblock {\em Annales Universitatis Scientarium Budapestinensis de Rolando
  Eötvös Nominatae Sectio Mathematica}, 12:147--155.

\bibitem[Yazıcıoğlu et~al., 2016]{Yazıcıoğlu2016}
Yazıcıoğlu, A., Abbas, W., and Egerstedt, M. (2016).
\newblock Graph distances and controllability of networks.
\newblock {\em IEEE Transactions on Automatic Control}, 61(12):4125--4130.

\bibitem[Yuan et~al., 2013]{Yuan2013}
Yuan, Z., Zhao, C., Di, Z., Wang, W., and Lai, Y. (2013).
\newblock Exact controllability of complex networks.
\newblock {\em Nature communications}, 4:2447.

\bibitem[Yueh, 2005]{Yueh2005}
Yueh, W. (2005).
\newblock Eigenvalues of several tridiagonal matrices.
\newblock {\em Applied Mathematics E-Notes}, 5:66--74.

\bibitem[Zhang and Zhou, 2017]{Zhang2017}
Zhang, Y. and Zhou, T. (2017).
\newblock Controllability analysis for a networked dynamic system with
  autonomous subsystems.
\newblock {\em IEEE Transactions on Automatic Control}, 62(7):3408--3415.

\end{thebibliography}
                                 % bibliography (preferred). The
                                 % correct style is generated by
                                 % Elsevier at the time of printing.

\end{document}